\documentclass{llncs}

\usepackage{misc, xspace, amsmath, amssymb}
\usepackage{boxedminipage}
\usepackage{tikz, tkz-graph, models, authblk}
\usetikzlibrary{arrows}

\newtheorem{nclaim}{Claim}

\title{Characterising Fixed Parameter Tractability of Query Evaluation Over Guarded TGDs}
\author{Cristina Feier}
\institute{University of Bremen}

\begin{document}
		\maketitle
	\begin{abstract}
		We study the parameterized complexity of evaluating Ontology Mediated Queries (OMQs) based on Guarded TGDs (GTGDs) and Unions of Conjunctive Queries (UCQs), in the case where relational symbols have unrestricted arity and where the parameter is the size of the OMQ. We establish exact criteria for fixed-parameter tractability (fpt) evaluation of recursively enumerable classes of such OMQs (under the widely held Exponential Time Hypothesis). One of the main technical tools introduced in the paper is an fpt-reduction from deciding parameterized uniform CSPs to parameterized OMQ evaluation.  The reduction preserves measures which are known to be essential for classifying recursively enumerable classes of parameterized uniform CSPs: submodular width (according to the well known result of Marx for unrestricted-arity schemas) and treewidth (according to the well known result of Grohe for bounded-arity schemas). As such, it can be employed to obtain hardness results for evaluation of recursively enumerable classes of parameterized OMQs both in the unrestricted and in the bounded arity case. Previously, in the case of bounded arity schemas, this has been tackled using a technique requiring full introspection into the construction employed by Grohe. \end{abstract}

	\newcommand{\mi}[1]{\mathit{#1}}
\newcommand{\ins}[1]{\mathbf{#1}}
\newcommand{\adom}[1]{\mathsf{dom}(#1)}
\renewcommand{\paragraph}[1]{\textbf{#1}}
\newcommand{\ra}{\rightarrow}
\newcommand{\fr}[1]{\mathsf{fr}(#1)}
\newcommand{\dep}{\Sigma}
\newcommand{\sch}[1]{\mathsf{sch}(#1)}
\newcommand{\ar}[1]{\mathsf{ar}(#1)}
\newcommand{\body}[1]{\mathsf{body}(#1)}
\newcommand{\head}[1]{\mathsf{head}(#1)}
\newcommand{\guard}[1]{\mathsf{guard}(#1)}
\newcommand{\class}[1]{\mathbb{#1}}
\newcommand{\size}[1]{||#1||}
\newcommand{\tw}[1]{\mathsf{TW}_{#1}}
\newcommand{\var}[1]{\mathsf{var}(#1)}
\newcommand{\omq}[1]{\mathsf{omq}(#1)}
\newcommand{\base}[1]{\mathsf{base}(#1)}
\newcommand{\chase}[2]{\mathsf{ch}_{#1}(#2)}
\newcommand{\ch}[3]{\mathsf{ch}_{#1}(#2,#3)}
\newcommand{\gchase}[2]{\mathsf{chase}_{\downarrow}(#1,#2)}
\newcommand{\lchase}[4]{\mathsf{chase}^{#1}_{#2}\left(#3,#4\right)}
\newcommand{\fmods}[2]{\mathsf{fmods}(#1,#2)}
\newcommand{\level}[2]{\mathsf{level}_{#1}(#2)}
\newcommand{\rew}[2]{\mathsf{rew}(#1,#2)}
\newcommand{\type}{\mathsf{type}}
\newcommand{\atoms}[1]{\mathsf{atoms}(#1)}
\newcommand{\complete}[2]{\mathsf{complete}(#1,#2)}
\newcommand{\N}[1]{\mathsf{N}(#1)}
\newcommand{\rt}[1]{\mathsf{root}(#1)}
\newcommand{\PTime}{\text{\rm \textsc{PTime}}}
\newcommand{\Pclass}{\text{\rm \textsc{P}}}
\newcommand{\NP}{\text{\rm \textsc{NP}}}
\newcommand{\FPT}{\text{\rm \textsc{FPT}}}
\newcommand{\W}{\text{\rm \textsc{W}[1]}}
\newcommand{\Wp}[1]{\text{\rm \textsc{W}[#1]}}
\newcommand{\WW}{\text{\rm \textsc{W}[2]}}
\newcommand{\EXP}{\text{\rm \textsc{ExpTime}}}
\newcommand{\TWOEXP}{\text{\rm \textsc{2ExpTime}}}
\newcommand{\THREEEXP}{\text{\rm \textsc{3ExpTime}}}
\newcommand{\FOUREXP}{\text{\rm \textsc{4ExpTime}}}
\newcommand{\logspace}{\text{\rm \textsc{LogSpace}}}


%
%
%

\newcommand{\markfull}{\qedboxfull}
\newcommand{\markempty}{\qed} 
	
\section{Introduction}
\label{sec:intro}

Ontology mediated querying refers to the scenario where queries are posed to a database enhanced with a logical theory, commonly referred to as an \emph{ontology}. The ontology refines the specific knowledge provided by the database by means of a logical theory. Popular ontology languages are decidable fragments of first order logic (FOL) like description logics \cite{DBLP:conf/dlog/2003handbook}, guarded TGDs \cite{CaGK13}, monadic disjunctive datalog \cite{BCLW14}, but not only --  non-monotonic formalisms 
like answer set programming \cite{ASPLifschitz99}, 
or combinations of languages from the former and latter category like r-hybrid knowledge bases \cite{Rosa05}, g-hybrid knowledge bases \cite{DBLP:journals/tplp/HeymansBPFN08}, etc. have also been considered. As concerns query languages, atomic queries (AQs), conjunctive queries (CQs), and unions thereof (UCQs) are commonly used. A tuple $(\mathcal{L}, \mathcal{Q})$, where $\mathcal{L}$ is an ontology language and $\mathcal{Q}$ is a query language, is referred to as an \emph{OMQ language}.


One thoroughly explored fragment of FOL as a basis for ontology specification languages is that of \emph{tuple generating dependencies} (TGDs). A tgd is a rule (logical implication) having as body and head conjunctions of atoms, where some variables occurring in head atoms might be existentially quantified (all other variables are universally quantified). As such, it potentially allows the derivation of atoms over fresh individuals (individuals not mentioned in the database). Answering (even atomic) queries with respect to sets of tgds is undecidable \cite{CaGK13}. However, there has been lots of work on identifying decidable fragments \cite{CaGK13,BLMS11,CaGP12}. A prominent such fragment is that of Guarded TGDs (GTGD) \cite{CaGK13}: a tgd is \emph{guarded} if all universally quantified variables occur as terms of some body atom, called \emph{guard}. 


While query answering with respect to GTGDs is decidable, the combined complexity of the problem is quite high: \exptime-complete for bounded arities schemas, and \TWOEXP-complete in general. A natural question is when can OMQs from (GTGD, UCQ) be evaluated efficiently? A first observation is that by fixing the set of tgds, the complexity drops to \NP\ for evaluating CQs, and to \PTime\ for evaluating AQs and CQs of bounded treewdith \cite{Cali:2009:DPU:1514894.1514897}, which is similar to the complexity of query evaluation over databases \cite{DBLP:conf/vldb/Yannakakis81}. 

Efficiency of query evaluation over databases has been a long evolving topic in the database community: starting with results concerning tractability of acyclic CQs evaluation \cite{DBLP:conf/vldb/Yannakakis81}, extended to  bounded treewidth CQs in \cite{ChRa00}, and culminating, in the case of bounded arity schemas, with a famous result of Grohe which characterizes classes of CQs which can be efficiently evaluated in a parameterized complexity framework where the parameter is the query size. In this setting, under the assumption that $\FPT \neq \W$,  Grohe \cite{Grohe07} establishes that those and only those recursively enumerable (r.e.) classes of CQs which have bounded treewidth modulo homomorphic equivalence are fixed-parameter tractable (fpt). It is also shown  that fpt coincides with polytime evaluation in this case.


As concerns OMQs from (GTGD, UCQ) over bounded arity schemas, a similar characterization to that of Grohe has been established in a parameterized setting where the parameter is the size of the OMQ \cite{BDFLP-PODS20}. The cut-off criterium for efficient evaluation is again bounded treewidth modulo equivalence, only this time equivalence takes into account also the ontology. An OMQ from (GTGD, UCQ) has \emph{semantic treewidth} $k$ if there exists an equivalent OMQ from (GTGD, UCQ) whose UCQ has syntactic treewidth $k$ \cite{BFGP19}. Then, under the assumption that $\FPT \neq \W$, a r.e. class of OMQs from (GTGD, UCQ) over bounded arity schemas can be evaluated in fpt iff it has bounded semantic treewidth. The similarity of the characterization with the database case is not coincidental: the results for OMQs build on the results of Grohe in a non-trivial way. In fact, the lower bound proof uses a central construction from \cite{Grohe07}, but has to employ sophisticated techniques to adapt this to OMQs.

\smallskip
The main open question which is adressed in this paper is: \emph{when is it possible to efficiently evaluate OMQs from (GTGD, UCQ) in the general case, i.e. when there is no restriction concerning schema arity?} We again consider a parameterized setting, where the parameter is the size of the OMQ. This is a reasonable choice as the size of the OMQ is usually much smaller than the size of the database. As such, we are interested in \emph{investigating the limits of fixed-parameter tractability of evaluating a class of OMQs $\class{Q}$ from (GTGD, UCQ)}. We denote such a parameterized problem as $\text{p-OMQ}(\class{Q})$. 
\smallskip

Before giving an overview of our main results, we review some results concerning efficiency of solving \emph{constraint satisfaction problems} (CSP). These are relevant, as CQ evaluation over databases is tightly linked to solving a certain type of CSPs, called \emph{uniform CSPs}. Given two classes of relational structures $\class{A}$ and $\class{B}$, a CSP problem  $(\class{A}, \class{B})$ asks whether there exists a homomorphism from some relational structure in $\class{A}$ to another relational structure in $\class{B}$. The uniform case refers to the situation where $\class{A}$ is fixed and $\class{B}$ is the class of all relational structures; in this case, the problem is denoted as $(\class{A}, \_)$.  The parameterized version of the problem (where for a problem instance $(A,B)$ with $A \in \class{A}$ the parameter is the size of $A$) is denoted as $\text{p-CSP}(\class{A}, \_)$. When restricted to classes of finite structures, uniform CSPs can be seen as an alternative presentation of the problem of evaluating a class of Boolean CQs over databases. In fact, Grohe's characterization for fpt evaluation of r. e. classes of CQs in the bounded arity case has been achieved via a uniform CSP detour \cite{Grohe07}. 


In the unrestricted arity case, Marx \cite{Marx10} established in a seminal result the border for fpt evaluation of uniform CSPs of the form $\text{p-CSP}(\class{A}, \_)$, where $\class{A}$ is closed under underlying hypergraphs. 
The restriction has been lifted in \cite{DBLP:conf/ijcai/ChenGLP20}, yielding a full characterization for parameterized uniform CSPs of unrestricted arity. Both results are based on a widely held conjecture, the Exponential Time Hypothesis \cite{DBLP:journals/jcss/ImpagliazzoPZ01} and rely on  a new structural measure, \emph{submodular width}: 

\begin{theorem}[Theorem 1, \cite{DBLP:conf/ijcai/ChenGLP20}]
	\label{thm:ChenUnbCSPs}
	Let $\class{C}$ be a r.e. class of structures. Assuming the Exponential Time Hypothesis, $\text{p-CSP}(\class{C}, \_)$ is fixed-parameter tractable iff  $\class{C}$ has bounded semantic submodular width. 
\end{theorem}

As our first main result shows, submodular width also plays a role in our characterization:



\smallskip
 \textbf{Main Result 1.} Let $\class{Q}$ be a r. e. class of OMQs from (GTGD, UCQ). Assuming the Exponential Time Hypothesis, $\text{p-OMQ}(\class{Q})$ is fixed-parameter tractable iff $\class{Q}$ has bounded semantic submodular width.
\smallskip

To prove the result, we exploit the fact that every OMQ from (GTGD, UCQ) can be rewritten into an OMQ from (GDLog, UCQ) \cite{BDFLP-PODS20}, where GDLog stands for Guarded Datalog, the restriction of GTGD to rules with only universally quantified variables. For OMQs from (GDLog, UCQ), we construct equivalent OMQs called \emph{covers} which are witnesses for bounded semantic submodular width. Covers are based on sets of \emph{characteristic databases for OMQs} which are databases that entail the OMQs and which are sufficiently minimal with respect to the homomorphism order, in a very specific sense. Typically, in a database setting, the database induced by a CQ (or its core) can be seen as a canonical database  which entails the query and on which to base further constructions. However, in the case of OMQs which pose restrictions  on the database schema this is no longer possible: a CQ might contain symbols which are not allowed to occur in a database. In fact, this is a typical usage of ontologies: to enrich the database schema with new terminology. 
Based on these notions, for OMQs from (GDLog, UCQ), we establish a syntactic characterization of the fpt border:

\smallskip
\textbf{Main Result 2.} For $\class{Q}$ a r. e. class of OMQs from (GDLog, UCQ), let $\class{Q}_c$ and $\class{D}_{\class{Q}}$ be the classes of covers and characteristic databases for OMQs from $\class{Q}$, respectively. Under the Exponential Time Hypothesis, the following statements are equivalent: 
\begin{enumerate}
	\item $\text{p-OMQ}(\class{Q})$ is fixed-parameter tractable;
	\item $\class{Q}_c$ has bounded submodular width;
	\item $\class{D}_{\class{Q}}$ has bounded submodular width. 
\end{enumerate}
\smallskip
The hardness result for the above characterization is obtained via an fpt-reduction from parameterized uniform CSP evaluation to parameterized OMQ evaluation.

\smallskip
\textbf{Main Result 3.} For $Q$ an OMQ from $(\text{GDLog}, \text{UCQ})$, there exists an fpt-reduction from $\text{p-CSP}(\class{D}_{Q},\_)$ to $\text{p-OMQ}(\{Q\})$.
\smallskip

The reduction is important also as a stand-alone result as it can be used as a black-box tool to port results from the uniform CSP/database realm to the OMQ one. Actually, both Main Result~1 and Main Result~2 above can be cast into results characterizing fixed-parameter tractability for classes of OMQs in the bounded arity case, by replacing submodular width with treewidth and the Exponential Time Hypothesis with the assumption that $\FPT \neq \W$. 
As such, it is possible to retrieve the semantic characterization from \cite{BDFLP-PODS20} without going into details of the construction employed by Grohe in \cite{Grohe07}, and also to provide an alternative syntactic characterization for fpt evaluation of OMQs from (GDLog, UCQ). 

\section{Preliminaries}
\label{sec:prelim}

\textbf{Structures, Databases.}
A \emph{schema} $\Sbf$ is a finite set of relation symbols with associated arities. An \emph{$\Sbf$-fact} has the form $r(\abf)$, where $r \in \Sbf$, and $\abf$ is a tuple of constants of size the arity of $r$. An \emph{$\Sbf$-structure} $A$ is a set of $\Sbf$-facts. The domain of a structure $A$, $\adom{A}$, is the set of constants which occur in facts in $A$. Given a structure $A$ and a subset $C \subseteq \adom{A}$, the \emph{sub-structure of $A$ induced by $C$}, $A|_C$,it is the structure containing all facts $r(\bbf) \in A$ such that $\bbf \subseteq C$. \emph{The product of two structures $A$ and $B$}, $A \times B$, is a structure with domain $\adom{A}  \times \adom{B}$ consisting of all facts of the form $r((a_1,b_1), \ldots, (a_n,b_n))$, where $r(a_1, \ldots, a_n) \in A$ and $r(b_1, \ldots, b_n) \in B$. 
Given two structures $A$ and $B$, a function $f:\adom{A} \to \adom{B}$ is said to be a \emph{homomorphism} from $A$ to $B$, if for every fact $r(\abf) \in A$, there exists a fact $r(\bbf) \in B$ such that $h(\abf)=\bbf$. The image of $A$ in $B$ under $f$, $f(A)$, is the set of facts of the form $r(f(\abf))$ in $B$, where $r(\abf)$ is from $A$. When such a homomorphism exists we say that $A$ \emph{maps into} $B$, denoted $A \rightarrow B$.



Two structures are \emph{equivalent} if $A \rightarrow B$ and $B \rightarrow A$. We write $A \leftrightarrow B$. They are \emph{isomorphic} if there exists a homomorphism $h$ from $A$ to $B$ which is bijective and onto, i.e. for every fact $r(\bbf) \in B$, there exists a fact $r(\abf) \in A$ such that $h(\abf)=\bbf$. A structure $A$ is a \emph{core} if every homomorphism from $A$ to itself is injective. Every structure $A$ has an induced sub-structure  $A'$ which is equivalent to $A$ and is a core. All cores of a structure are isomorphic. The homomorphism relation $\rightarrow$ is a pre-order over the set of all structures. When restricted to structures which are cores (taken up to isomorphism), $\rightarrow$ is a partial order. Following~\cite{foniok2007homomorphisms}, we will refer to this order as  the \emph{homomorphism order}.

An \emph{$\Sbf$-database} is a finite $\Sbf$-structure. For $D$ a database, and $\abf \subseteq \adom{D}$, $\abf$ is a \emph{guarded set  in $D$} if there exists a fact $r(\abf')$ in $D$ such that $\abf \subseteq \abf'$. It is a \emph{maximal guarded set} if there exists no strictly guarded superset. As in the previous definitions, we will sometimes abuse notation by using tuples of constants to refer to the underlying sets of constants instead.

\smallskip

\textbf{Conjunctive Queries, Atomic Queries.}
A \emph{conjunctive query (CQ)} is a formula of the form $q(\xbf)=\exists \ybf \phi(\xbf, \ybf)$, with $\xbf$ and $\ybf$  tuples of variables and $\phi(\xbf, \ybf)$ a conjunction of atoms having as terms only variables from $\xbf \cup \ybf$. The set $\xbf$ is the set of \emph{answer variables} of $q$, while the set $\ybf$ is the set of \emph{existential variables} of $q$. We  denote with $\var{q}$ the set of variables of $q$, and with $D[q]$ the \emph{canonical database of $q$}, i.e. the set of atoms which occur in $\phi$ (viewed as facts). When $\xbf$ is empty, the CQ is said to be \emph{Boolean} (BCQ). We will sometimes use BCQs or their canonical databases interchangeably. A sub-query of a BCQ $q$ is a BCQ $p$ such that $D[p] \subseteq D[q]$. A \emph{union of conjunctive queries (UCQ)} is a formula of the form $q(\xbf)=q_1(\xbf) \vee \ldots q_n(\xbf)$, where each $q_i(\xbf)$ is a CQ, for $i \in [n]$. An \emph{atomic query (AQ)} is a CQ in which $\phi$ contains a single atom. In the following, whenever we refer to CQs or UCQs, we tacitly assume they are Boolean. As concerns AQs, unless stated otherwise, we assume they are of the form $r(\xbf)$, i.e. they contain no existentially quantified variables. 

For a structure $I$, a CQ $q(\xbf)$, and a tuple of constants $\abf$ from $\adom{I}$, \emph{$\abf$ is an answer to $q$ over $I$}, or $I \models q(\abf)$, if there is a homomorphism $h$ from $D[q]$ to $I$ such that $h(\xbf)=\abf$. If $q(\xbf)$ is a UCQ of the form $q_1(\xbf) \vee \ldots q_n(\xbf)$, $I \models q(\abf)$ if $I \models q_i(\abf)$, for some $i \in [n]$.

\smallskip
\textbf{Ontology Mediated Queries.} 
An \emph{ontology mediated query (OMQ)} $Q$ is a triple $(\Omc, \Sbf, q(\xbf))$, where $\Omc$ is an ontology, $\Sbf$ is a schema, and $q(\xbf)$ is a query.  When $\Omc$ is specified using the ontology language $\mathcal{L}$, and $q$ using the query language $\mathcal{Q}$, we say that $Q$ belongs to the OMQ language $(\mathcal{L}, \mathcal{Q})$. The schema $\Sbf$ specifies which relational symbols can occur in databases over which $Q$ is evaluated. We say that $Q$ is an \emph{$\Sbf$-OMQ}.

Given an $\Sbf$-database $D$, and a tuple of constants $\abf$, all of which are from $\adom{D}$, we say that $\abf$ is an \emph{answer to $Q$ over $D$}, or $D \models Q(\abf)$,  if $\Omc \cup D \models q(\abf)$,  where $\models$ is the entailment relation in $\Lmc$. For $Q_1$ and $Q_2$ two OMQs over the same schema $\Sbf$, we say that $Q_1$ is \emph{contained in} by $Q_2$, written $Q_1 \subseteq Q_2$, if for every $\Sbf$-database $D$ and tuple of constants $\abf$: $D \models Q_1(\abf)$ implies $D \models Q_2(\abf)$. We also say that $Q_1$ is \emph{equivalent} to $Q_2$ if $Q_1 \subseteq Q_2$ and $Q_2 \subseteq Q_1$. 
 

\smallskip

\textbf{TGDs, Guarded TGDs.} 
Tuple Generatings Dependencies (TGDs) are first order sentences of the form $\forall \xbf \forall \ybf \ \phi(\xbf, \ybf) \to \exists \zbf\ \psi(\xbf,\zbf)$, with $\phi$ and $\psi$ conjunctions of atoms having as terms only variables from $\xbf \cup \ybf$, and from $\xbf \cup \zbf$, respectively. Such a sentence will be abbreviated as $\phi(\xbf, \ybf) \to \exists \zbf\ \psi(\xbf,\zbf)$. The problem of \emph{evaluating a UCQ $q(\xbf)$ over a set of TGDs $\Omc$ w.r.t. an  $\Sbf$-database $D$} consists in checking whether for some tuple $\abf$ over $\adom{D}$, it is the case that $D \models Q(\abf)$, where $Q$ is the OMQ $(\Omc, \Sbf, q(\xbf))$. While the problem is undecidable,  
\cite{CaGK13}, it can be characterized via a completion of the database $D$ to a structure which is a universal model of the set of TGDs $\Omc$ and $D$, called \emph{chase} \cite{MaMS79,DeNR08,JoKl84}. We will denote with $\chase{\Omc}{D}$ the chase of $\Omc$ w.r.t. $D$. Then, $D \models Q(\abf)$ iff $\chase{\Omc}{D} \models q(\abf)$. 

There are several variants of the chase; here we describe the  \emph{oblivious chase}. Let $(\ch{k}{\Omc}{D})_{k \geq 0}$ be a sequence of structures such that $\ch{0}{\Omc}{D}=D$. Then, for every $i>0$, $\ch{i}{\Omc}{D}$ is obtained from $\ch{i-1}{\Omc}{D}$ by considering all homomorphisms $h$ from the body of some tgd $\phi(\xbf, \ybf) \to \exists \zbf \psi(\xbf, \zbf)$ in $\Omc$ to $\ch{i-1}{\Omc}{D}$ s.t. at least one atom from $\phi$ is mapped by $h$ into a fact from $\ch{i-1}{\Omc}{D} \setminus \ch{i-2}{\Omc}{D}$, and adding to $\ch{i}{\Omc}{D}$ all facts obtained from atoms in $\psi(\xbf, \zbf)$ by replacing each $x \in \xbf$ with $h(x)$ and each $z \in \zbf$ with some fresh constant. Then, $\chase{\Omc}{D}=\bigcup_{k \geq 0} \ch{k}{\Omc}{D}$. Note that $\chase{\Omc}{D}$ might be infinite. 

%

A TGD with a body atom which has as terms all its universally quantified variables is said to be \emph{guarded}. The language of guarded TGDs will be denoted as GTGD. Unlike evaluation of OMQs based on unrestricted TGDs, evaluation of OMQs from (GTGD, UCQ) is decidable \cite{journals/tods/BourhisMMP16}. 
 By further restricting guarded tgds to rules with universally quantified variables only, one obtains the ontology language \emph{guarded Datalog (GDLog)}. For every GDLog ontology $\Omc$ and every database $D$, $\chase{\Omc}{D}$ is finite and, furthermore for every fact $r(\abf) \in \chase{\Omc}{D}$, there exists a guarded set $\abf'$ over $D$ such that $\abf \subseteq \abf'$. 

\medskip

\textbf{Parameterized Complexity.}
For $\Sigma$ some finite alphabet, a \emph{parameterized problem} is a tuple $(P,\kappa)$, where $P \subseteq \Sigma^{*}$ is a problem, and $\kappa: \Sigma^* \to \mathbb{N}$ is a \PTime~computable function called the \emph{parameterization} of $P$. Such a parameterized problem is \emph{fixed-parameter tractable} if there exists an algorithm for deciding $P$ for an input $x \in  \Sigma^*$ in time $f(\kappa(x))poly(|x|)$, where $f$ is a computable function and $poly$ is a polynomial. The class of all fixed-parameter tractable problems is denoted as \FPT. 

Given two parameterized problems $(P_1, \kappa_1)$ and $(P_2, \kappa_2)$ over alphabets $\Sigma_1$ and $\Sigma_2$, an \emph{fpt-reduction} from $(P_1, \kappa_1)$ to $(P_2, \kappa_2)$ is a function $R:\Sigma^*_1 \to \Sigma^*_2$ with the following properties: 
\begin{enumerate}
	\item $x \in P_1$ iff  $R(x) \in P_2$, for every $x \in \Sigma^*_1$, 
	\item there exists a computable function $f$ such that $R(x)$ is computable in time $f(\kappa_1(x))poly(|x|)$, 
	\item there exists a computable function $g$ such that $\kappa_2(R(x)) \leq g(\kappa_1(x))$, for all $x \in \Sigma^*_1$. 
\end{enumerate}
Downey and Fellows \cite{DoFe95I} defined a hierarchy of parameterized complexity classes $\Wp{0} \subseteq \Wp{1} \subseteq \Wp{2} \dots$, where $\Wp{0}=\FPT$ and each inclusion is believed to be strict. Each class $\Wp{i}$, with $i \geq 0$, is closed under fpt-reductions. 

A class of interest for us is $\Wp{1}$ as under the assumption that $\FPT \neq \Wp{1}$, it is possible to establish intractability results (non-membership to \FPT)\ for parameterized problems. A well-known $\Wp{1}$-complete problem is the parameterized $k$-clique problem, where the parameter is $k$: for an input $(G, k)$, with $G$ a graph and $k \in \mathbb{N}^*$ it asks whether $G$ has a $k$-clique. An even stronger assumption than $\FPT \neq \Wp{1}$ is the \emph{Exponential Time Hypothesis}: it states that 3-SAT with $n$ variables cannot be decided in $2^{o(n)}$ time \cite{DBLP:journals/jcss/ImpagliazzoPZ01}. The assumption is standard in the parameterized complexity theory and can be used as well to establish intractability results. 

\medskip
\textbf{Structural Measures: Treewdith, Submodular Width.}
A \emph{hypergraph} is a pair $H=(V,E)$ with $V$ a set of \emph{nodes} and $E \subseteq 2^V \setminus \{ \emptyset \}$ a set of \emph{edges}.  A \emph{tree decomposition} of $H$ is a pair $\delta = (T_\delta, \chi)$, with $T_\delta = (V_\delta,E_\delta)$ a tree, and $\chi$ a labeling function $V_\delta \ra 2^{V}$ such that: 
\begin{enumerate}
	
	\item $\bigcup_{t \in V_\delta} \chi(t) = V$.
	
	\item If $e \in E$, then $e \subseteq \chi(t)$ for some $t \in V_\delta$. 
	
	\item For each $v \in V$, the set of nodes $\{t \in V_\delta \mid v \in \chi(t)\}$ induces a connected subtree of $T_\delta$.
\end{enumerate}
The \emph{treewidth} of $H$, $\mn{TW}(H)$, is the smallest $k$ such that there exists a tree decomposition $(T_\delta,\chi)$ of $H$, with $T_\delta=(V_\delta,E_\delta)$, such that for every $t \in V_\delta$, $|\chi(t)| \leq k$. A function $f: 2^V \rightarrow \mathbb{R}_{\geq 0}$ is \emph{submodular} if $f(X)+f(Y) \geq f(X \cap Y) +f(X\cup Y)$. 
It is \emph{edge-dominated} if  $f(e) \leq 1$ for all \mbox{$e \in E$} .
The \emph{submodular width} of $H$, $\mn{SMW}(H)$, is the smallest $k$ such that for every monotone submodular edge-dominated function $f$, for which $f(\emptyset)=0$, there exists a tree decomposition $(T_\delta,\chi)$ of $H$, with $T_\delta=(V_\delta,E_\delta)$, such that $f(\chi(t)) \leq k$ for all $t \in V_\delta$. 

Every relational structure $I$ has an associated hypergraph $H_I$ whose vertices are the constants of $I$ and whose edges are the sets of constants from tuples $\abf$, for every atom $R(\abf)$ of $I$. As such, the structural measures on hypergraphs can be lifted to relational structures, CQs, UCQs, and OMQs. Let $\mn{X}$ range over $\{\mn{TW}, \mn{SMW}\}$. 
 Then, for a relational structure $I$, $\mn{X}(I)=\mn{X}(H_I)$. For a CQ $q$, $\mn{X}(q)= \mn{X}(H_q)$, while for a UCQ $q'$, $\mn{X}(q')=\mn{max}_{q \mbox{ is a CQ in }q'} (X(q))$. Finally, for an OMQ $Q=(\Omc, \Sbf, q)$, with $q$ a UCQ, $\mn{X}(Q)=\mn{X}(q)$. 

\section{Normalizing OMQs: Characteristic Databases and Covers}
\label{sec:cover}

As we have seen in the Introduction, \emph{equivalence-based measures} frequently play a role in characterizations of efficient (fpt) evaluation. Following \cite{DBLP:journals/sigmod/BarceloPR17,BFLP19,BDFLP-PODS20}, we refer to such measures as \emph{semantic measures}. In particular, for an OMQ $Q \in (\mathcal{L}, \text{UCQ})$, with $\mathcal{L}$ an ontology language, and a structural measure $\mn{X} \in \{\mn{TW}, \mn{SMW}\}$, 
the semantic $\mn{X}$-width $k$ of $Q$ is the smallest $k$ such that there exists an OMQ $Q'$ from  $(\mathcal{L}, \text{UCQ})$ with $Q \equiv Q'$ and $\mn{X}(Q')=k$.  A class $\class{Q}$ of OMQs has \emph{bounded semantic $\mn{X}$-width} if there exists some $k>0$ such that every OMQ in $\class{Q}$ has semantic $\mn{X}$-width at most $k$. However, in order to establish such characterizations, an important issue is finding witnesses of (bounded) semantic measures, i.e. problems of low syntactical measures which are equivalent to the original ones.

In the case of CQs, cores serve as witnesses for semantic treewidth \cite{Grohe07} and also for semantic submodular width \cite{DBLP:conf/ijcai/ChenGLP20}. Thus, as concerns classes of CQs of bounded semantic $\mn{X}$-width, with $\mn{X} \in \{\mn{TW}, \mn{SMW}\}$, the class of cores of CQs from the original class serves as a witness, i.e. it has actual bounded $\mn{X}$-width. 


However, for OMQs based on UCQs, resorting to cores of CQs in UCQs does not necessarily lead to witnesses of low width. The ontology also plays a role in lowering semantic measures. For examples of this phenomenon 
as concerns semantic treewidth, see \cite{BFLP19,BDFLP-PODS20}. Here, we show how the ontology influences semantic submodular width: 

\begin{example}
	\label{ex:cover}
	For $R$ a binary relational symbol, $i \in \mathbb{N}$, with $i>1$, and $\xbf_i$ an $i$-tuple of variables, we denote with $\psi_i^R(\xbf_i)$ the formula $R(x_1,x_2) \wedge R(x_1, x_3) \wedge \dots \wedge R(x_{i-1}, x_i)$, i.e. the hyper-graph associated to $\psi_i^R$ is the $i$-clique. Let $T$ be a binary relational symbol and $\class{Q}$ be the class of OMQs $(Q_i)_{i>1}$, with $Q_i=(\Omc_i, \Sbf_i, q_i)$, where: 
\begin{center}
\begin{tabular}{lllll}
	$\Omc_i=\{S_i(\xbf_i) \to \psi_i^{R}(\xbf_i)\}$ &&
	$\Sbf_i=\{S_i, T\}$ &&
	$q_i=\exists \xbf_i\ \psi_i^{R}(\xbf_i) \wedge  \psi_i^{T}(\xbf_i)$
\end{tabular}
\end{center}	
Then $\class{Q}$ has unbounded submodular width. As for every $i>1$, $H_{D[q_i]}$ is the $i$-clique, every tree decomposition $(T, \chi)$ of $H_{D[q_i]}$, with $T=(V, E)$, must contain some node $t \in V$ such that $\chi(t)=\xbf_i$. Let  $f:2^{\xbf_i} \to \mathbb{R}_{\geq 0}$ be the monotone submodular function $f(X)=|X|/2$. Then $f$ is also edge-dominated with respect to $D[q_i]$, and its minimum over all tree decompositions of $H_{D[q_i]}$ is $i/2$. Thus, $\mn{SMW}(Q_i) \geq i/2$, for every $i>1$.

On the other hand, for every $i >1$, $R$ is not part of the schema $\Sbf_i$ and the only way to derive it is using the unique tgd from $\Omc_i$. Thus, every $\Sbf_i$-database $D_i$ such that $D_i \models Q_i$ must contain an atom of the form $S_i(\cbf_i)$, where $\cbf_i$ is an $i$-tuple of constants. Then, for every $i>0$, $Q_i$ is equivalent to the OMQ $Q'_i=(\Omc_i, \Sbf_i,q'_i)$, with $q'_i=\exists \xbf_i\ S_i(\xbf_i) \wedge \psi_i^{R}(\xbf_i) \wedge \psi_i^{T}(\xbf_i)$.

 Let $\class{Q}'$ be the class of OMQs $(Q'_i)_{i>1}$. As for every $i>1$, $q'_i$ is guarded, i.e. it contains an atom $S_i(\xbf_i)$ which has as terms $\var{q'_i}$, it follows that $\mn{SMW}(q'_i) \leq 1$: this is due to the fact that only edge-dominated functions are considered when defining the submodular width and thus the guard ensures that for every such function $f$, $f(\xbf_i) \leq 1$. Thus, $\class{Q}'$ has bounded submodular width and $\class{Q}$ has bounded semantic submodular width.  
	
\end{example}

Example~\ref{ex:cover} shows how submodular width can be lowered by adding extra atoms to CQs in the original OMQ. This cannot not the case for treewidth, as it is a monotonic measure. In this section, we normalize OMQs from (GDLog, UCQ) by extending (images of) CQs in the original OMQ with facts occurring in databases which entail the OMQ. We obtain equivalent OMQs called \emph{covers}. The purpose of the added atoms is to provide guards for atoms in CQs and to potentially lower submodular width, as in Example \ref{ex:cover}. Intuitively, by adding such guards, the submodular width is decreased as the space of edge-dominated submodular functions is shrunk. 

At the same time, when constructing covers, we do not want to add too many atoms: adding cliques of unbounded size  (without a guard) would obviously not decrease submodular width. As such, we will use as the basis for the construction only certain databases which entail an OMQ, which are small w.r.t. the homomorphism order and which we call \emph{(extended) characteristic databases}. As next example shows, for OMQs with restricted database schemas, there might be no minimal databases w.r.t. the homomorphism order which entail the OMQ. This is a side-effect of the fact that the homomorphism order over (core) databases which entail an OMQ is in general not well-founded, and the fact that due to schema restrictions, the canonical databases of CQs occurring in the OMQ might not be valid databases.

\begin{example}
	\label{ex:notwellfounded}
	
	Let $Q=(\Omc, \Sbf, q)$ be the following OMQ from (GDLog, UCQ): 
\begin{align*}
\Omc&=\{A(x), R(x,y) \to B(y), B(x), R(x,y) \to B(y)\} \\
\Sbf&=\{A, C\} \\
q&=\exists x \ B(x) \wedge R(x,x) \wedge C(x)
\end{align*}
	

Then, for every $n \in \mathbf{N}$, let $D_n$ be the $\Sbf$-database: 

$$\{A(x_0), C(x_{n}), R(x_n,x_n)\} \cup \{R(x_i, R(x_{i+1})) \mid 0 \leq i < n\}$$

It can be verified that for every $n \in \mathbf{N}$: $D_n \models Q$, $D_n$ is a core, and $D_{n+1} \to D_n$. Thus, the homomorphism order over core structures which entail $Q$ is not well-founded. Furthermore, every $\Sbf$-database $D$ which entails $Q$ has the property that there exists some $n \in \mathbf{N}$ such that $D_n \to D$. Thus, $D$ is not minimal w.r.t. $\to$. This contrasts to the case where the database schema is \emph{full}, i.e. there are no restrictions on the relational symbols which migth occur in a database: in that scenario we could identify such a minimal database simply by considering $D[q]$. 
	
\end{example}

As explained earlier, we will identify a class of databases which is good enough for our purposes (small w.r.t. the homomorphism order in a certain sense) by successive refinements of the set of databases which entail an OMQ. These refinement steps are described in  Sections~\ref{subsec:qi},~\ref{subsec:unravel}, and~\ref{subsec:diver}. Section~\ref{subsec:charac} brings all these concepts together to define (extended) characteristic databases and covers.

\subsection{Query Initial Databases}
\label{subsec:qi} 

In our quest to refine databases which entail a query we start by considering how CQs map into the chase of a given database. In the following, a \emph{contraction} of a CQ $q$ is a CQ obtained from $q$ by variable identification. For a CQ $q$ and an instance $I$, we write $q \rightarrow^{io} I$ if $q \rightarrow I$ and every homomorphism from $q$ to $I$ is injective.

\begin{definition}
	\label{def:qi-initial}
	For $Q=(\Omc, \Sbf, q)$ from $(\text{GDLog}, \text{UCQ})$ and $D$ an $\Sbf$-database such that $D \models Q$, we say that $D$ is \emph{query-initial (qi) w.r.t. $Q$} if for every \Sbf-database $D'$ such that $D' \rightarrow D$ and $D' \models Q$, and every contraction $p$ of some CQ in $q$ it is the case that: $p \rightarrow^{io} \chase{\Omc}{D'}$ iff $p \rightarrow^{io} \chase{\Omc}{D}$.
\end{definition}

Thus, qi databases entail an OMQ and are minimal w.r.t. the set of contractions which map injectively only into their chase. They are related to injectively only databases which have been introduced in \cite{BFLP19}. By a simple induction argument, it can be shown that:

\begin{lemma}
	\label{lem:io-prop}
	Let $Q$ be an $\Sbf$-OMQ from $(\text{GDLog}, \text{UCQ})$, and $D$  an $\Sbf$-database  with $D \models Q$. If $D$ is qi w.r.t. $Q$, then for every database $D' \to D$ such that $D' \models Q$, $D'$ is also qi w.r.t. $Q$. Otherwise, there exists a database $D' \to D$ such that $D' \models Q$ and $D'$ is qi w.r.t. $Q$.
\end{lemma}

\begin{example}
	\label{ex:ioDB}
	Let $Q=(\Omc, \Sbf, q)$ be the OMQ with: 
	\begin{align*}
	\Omc&=\{U(x,y,z) \wedge V(x,z) \to T(x,z), W(x,y,z) \to S(y,z)\} \\
	\Sbf&=\{R,U,V, W\}; \mbox{             } 
	q=\exists x,y,z\ R(x,y) \wedge S(y,z) \wedge T(z,x)
	\end{align*}
	
	Also, let $D_1$ and $D_2$ be the two $\Sbf$-databases: 
	$$D_1=\{R(a,b),W(d,b,a), U(a,d,a), V(a,a)\}$$	
	$$D_2=\{R(a,b),W(d,b,c), U(c,d,a),V(c,a) \}$$
	It can be checked that $D_2 \rightarrow D_1$.  Let $q'$ be the contraction of $q$ obtained by identification of $x$ and $z$: $$q'=\exists x,y \ R(x,y) \wedge S(y,x)\wedge T(x,x)$$ 
	Then $q' \rightarrow^{io} \chase{\Omc}{D_1}$, but $q' \not \rightarrow \chase{\Omc}{D_2}$. Thus, $D_1$ is not qi w.r.t. $Q$. However the only contraction which maps (injectively only) into $\chase{\Omc}{D_2}$ is $q$. Thus, $D_2$ is qi w.r.t. $Q$. 
	

%
\end{example}

\subsection{Guarded Unravelings}
\label{subsec:unravel}

Another concept which will be useful for defining characteristic databases is that of \emph{guarded unraveling}. The operation was first introduced in \cite{Gr99b}. A definition closer to our purposes is provided in \cite{BDFLP-PODS20}. For completeness, we repeat it here.

For an \Sbf-database~$D$ and a guarded set~$\abf$ over $D$, we construct a structure $I^{\abf}$, the \emph{guarded unraveling of $D$ at $\abf$}, in parallel with a tree decomposition $(T, \chi)$ of $I^{\abf}$, with $T=(V, E)$ as follows. $V$ is the set of all sequences of the form $\abf_0 \ldots \abf_n$, where $\abf_0=\abf$, and $\abf_1, \ldots, \abf_n$ are maximal guarded sets in $D$ such that $\abf_i \cap \abf_{i+1} \neq \emptyset$, and $\abf_i \neq \abf_{i+1}$, for every $0 \leq i <n$. For such a sequence $s$, we denote with $\mn{tail}(s)$ its last component $\abf_n$. For every $s_1, s_2 \in V$, it is the case that $(s_1, s_2) \in E$ iff $s_2=s_1\bbf$, for some maximal guarded set $\bbf$ in $D$. We further consider a set of constants $S$ such that $\abf \subseteq S$ and such that $S$ contains an infinite amount of copies for every constant $a \in \adom{D}$. We define $\chi:V \to 2^S$ in parallel with $I^{\abf}$ inductively. 

We start by setting $\chi(\abf)=\abf$ and initializing $I^{\abf}$ as $D|_{\abf}$. Then, for every node $t \in V$, such that $\chi(t)$ is defined, and for every $t' \in V$ such that $(t,t') \in E$, let $\mn{tail}(t)=\bbf_1$ and $\mn{tail}(t')=\bbf_2$. Also let $\bbf'$  be the set of constants obtained from $\bbf_2$ by replacing each constant $b \in \bbf_2 \setminus \bbf_1$ with a fresh copy from $S$ and 
let $D'$ be a copy of the database $D|_{\bbf_2}$ in which the constants from $\bbf_2$ have been replaced with their counterparts from $\bbf'$. Set $\chi(t')=\bbf'$ and update $I^{\abf}$ as $I^{\abf} \cup D'$. Note that $\bbf'$ is a guarded set in $D'$ (as $\bbf_2$ is a guarded set in $D|_{\bbf_2}$) and thus also in $I^{\abf}$. Thus, for  each $t \in V$, $\chi(t)$ is a guarded set. We say that the tree decomposition $(T, \chi)$ is a \emph{guarded tree decomposition}. 

Guarded unravelings have the property that for every OMQ  $Q'$ from (GDLog, AQ), every database $D$, every guarded set $\abf$ in $D$, and every tuple $\abf' \subseteq \abf$: $D \models Q'(\abf')$ implies $I^{\abf} \models Q'(\abf')$. It can also easily be seen that, due to the existence of a guarded tree decomposition, the submodular width of $I^{\abf}$ is 1. As such, for a given OMQ $Q=(\Omc, \Sbf, q)$ from (GDLog, CQ) and some $\Sbf$-database $D$ such that $D \models Q$, we will use them to disentangle parts of $D$ which are needed to entail specific atoms from some CQ in $q$, i.e. we will replace some parts of $D$ with corresponding guarded unravelings. 

In general, for a guarded set $\abf$ in $D$, $I^{\abf}$ is infinite and thus cannot be used straightaway. By compactness, there exists a finite subset (a database) $D^{\abf}$ which fulfills the same property. As we want to unravel databases which entail OMQs as much as possible, we place a stronger requirement on $D^\abf$: for every Boolean sub-query $p'$ of some CQ $p$ in $q$, whenever $I^{\abf} \models (\Omc, \Sbf, p')$ we require that $D^{\abf} \models (\Omc, \Sbf, p')$. As this requirement is relative to $Q$, we will refer to  $D^\abf$ as the \emph{guarded unraveling of $D$ at $\abf$ w.r.t. $Q$}. However, $Q$ will be clear in most cases from the context so it will be omitted.

\subsection{Diversifications}
\label{subsec:diver}

A \emph{diversification} of a database is essentially a database which maps into the original one in a specific way. For a function $f$ and $A$ a subset of its domain, we denote with $f|_A$ the restriction of $f$ on $A$. Also, we denote with  $\mn{ran}(f)$ the range of $f$. A homomorphism $h$ from a structure $A$ to a structure $B$ is said to be \emph{injective on guarded sets (i.g.s.)} if $h|_\abf$ is injective, for every guarded set $\abf$ in $A$. For a database $D$, a constant $c \in \adom{D}$ is \emph{isolated in $D$} if it occurs in a single fact in $D$. The \emph{kernel} of a database $D$, $\mn{ker}(D)$, is the set of non-isolated constants in $D$.

\begin{definition}
A \emph{diversification} of a database $D_0$ is a tuple $(D,\uparrow)$, where $D$ is a database which maps into $D_0$ via the homomorphism $\uparrow$ which is i.g.s. and for which $\uparrow_{|\mn{ker}(D)}$ is injective. We write $D \preceq D_0$, whenever there exists a diversification $(D, \uparrow)$ of $D_0$.


\end{definition}

We observe that for every database $D \subseteq D_0$, $(D, \uparrow)$ is a diversification of $D_0$, when $\uparrow$ is the identity function on $\adom{D}$. It can also be easily checked that $\preceq$ is transitive: for databases $D_0$, $D_1$, and $D_2$ and functions $\uparrow$ and $\downarrow$ such that $(D_1, \uparrow)$ is a diversification of $D_0$ and $(D_2, \downarrow)$ is a diversification of $D_1$, $(D_2, \uparrow \circ \downarrow)$ is a diversification of $D_0$, and thus $D_2 \preceq D_0$. 

\begin{example}
	\label{ex:div}
	
	Let $Q=(\Omc, \Sbf, q)$ and $D_2$ be as in Example~\ref{ex:ioDB}. Also, let $D$ be the $\Sbf$-database: $\{R(a,b), W(d,b,c), V(c,a)\}$ and $\mn{id}$ be the identity mapping on $\adom{D_2}$. As $\mn{ker}(D)=\{a,b,c\}$ and $\mn{id}$ is i.g.s., it follows that $(D, \mn{id})$ is a diversification of $D_2$ and thus $D \preceq D_2$. 
\end{example}

For a database $D$ which entails an OMQ $Q$, we will use diversifications in conjunction with guarded unravelings to disentangle $D$, i.e. to obtain databases $D'$ which are smaller w.r.t. the homomorphism order and which still entail $Q$. We start with a construction which glues guarded unravelings of some database $D_0$ to a database $D$ which maps into $D_0$ via some homomorphism $\uparrow$ which is i.g.s. (note that in this case $(D, \uparrow)$ might not be a diversification of $D_0$). In this case, we denote with $\mn{ext}_Q(D, \uparrow, D_0)$ the database obtained from $D$ by adding for each maximal guarded set $\abf$ in $D$ the database $D_0^{\abf}$ obtained from the guarded unraveling $D_0^{\uparrow(\abf)}$ of $D_0$ at $\uparrow(\abf)$ w.r.t. $Q$, by renaming the constants in ${\uparrow(\abf)}$ to those in~$\abf$. When $Q$ is clear from the context, we will write $\mn{ext}(D, \uparrow, D_0)$.



\begin{definition}
	For $Q$ an $\Sbf$-OMQ from (GDLog, UCQ) and $D_0$ an \Sbf-database with $D_0 \models Q$, $\mn{div}(D_0, Q)$ is the set of diversifications $(D, \uparrow)$ of $D_0$ for which $\mn{ext}(D, \uparrow, D_0) \models Q$. 
	
	A diversification $(D, \uparrow)$ from $\mn{div}(D_0, Q)$ is said to be \emph{minimal w.r.t. $Q$} if $D$ is a core and there is no other diversification $(D', \downarrow)$ from $\mn{div}(D_0, Q)$ such that $D' \preceq D$, and $D \not \preceq D'$. The set of all minimal diversifications of $D_0$ w.r.t. $Q$ is denoted as $\mn{mdiv}(D_0, Q)$. 

	\end{definition}

Intuitively, for a minimal diversification $(D, \uparrow)$ of $D_0$ w.r.t. $Q$ and the ensuing database $\mn{ext}(D, \uparrow, D_0)$, the $D$-part of $\mn{ext}(D, \uparrow, D_0)$ is important for preserving some part of the underlying structure (hypergraph) of $D_0$ needed to entail a CQ. The guarded unravelings which were added to $\mn{ext}(D, \uparrow, D_0)$ provide the necessary information to entail individual atoms in the query.

\begin{example}
	\label{ex:unravel}
	
	Let $Q=(\Omc, \Sbf, q)$ and $D_2$ be as in Example~\ref{ex:ioDB}. Also let $(D, \mn{id})$ be the diversification of $D_2$ introduced in Example~\ref{ex:div} and let $D^+=\mn{ext}(D, \mn{id}, D_2)$ be the database obtained from $D$ by adding guarded unravelings of $D_2$. We do not explicitly construct the guarded unravelings, but note that $D_2^{(a,c)}$ will contain a fact of the form $U(c, d', a)$, where $d'$ is a fresh copy of $d$. Let $D'=D \cup \{U(c, d', a)\}$. Then $D' \subseteq D^+$ and $D' \models Q$, thus $D^+ \models Q$. Thus, $D \in \mn{div}(D_2, Q)$. 
	
	On the other hand, it can be verified that for every database $D_3$ such that $D_3 \preceq D$, but $D \not \preceq D_3$, there exists no homomorphism $\downarrow$ from $D_3$ to $D_2$ such that $(D_3, \downarrow) \in \mn{div}(D_2,Q)$ (in other words there is no way to add guarded unravelings of $D_2$ to $D_3$ and still entail $Q$ if $D \not \preceq D_3$). To see why this is the case, observe that $D_3 \preceq D$ and $D \not \preceq D_3$, implies that $D_3$ is obtained from $D$ either by dropping facts or renaming some constant in an existing fact which was originally non-isolated (disjoining facts). Thus,  $(D, \uparrow) \in \mn{mdiv}(D_2, Q)$. 
	
	We look now at how $q$ maps into $\chase{\Omc}{D^+}$. There exists a homomorphism $h$ which maps $(x,y,z)$ into $(a,b,c)$. Thus, $h$ maps every guarded set in $q$ into some guarded set in $D$ -- in other words, $H_D$, the hypergraph of $D$, can be seen as a skeleton for $q$. At the same time, the atoms added by the guarded unravelings (in particular $U(c, d', a)$) allow the entailment of particular atoms from $q$ (like $T(x,z)$). 
	\end{example}


The following technical lemma shows how given a homomorphism $h$ from a CQ $p$ in an OMQ $Q$ to the chase of a database of the form $\mn{ext}(D, \uparrow, D_0)$, it is possible to construct a diversification of $D$, $(D', \downarrow)$, such that the database obtained by extending $D'$ with guarded unravelings of $D_0$ according to the composition homomorphism $\uparrow \circ \downarrow$ from $D'$ to $D$ still entails the OMQ $Q$.

\begin{lemma}
	\label{lem:divers-constr}
	Let $Q=(\Omc, \Sbf, q)$ be an OMQ from $(\text{GDLog}, \text{UCQ})$, $D$ and $D_0$ be  $\Sbf$-databases, and $\uparrow$ a homomorphism from $D$ to $D_0$ which is i.g.s. such that  $\mn{ext}(D, \uparrow, D_0) \models Q$. Also, let $h$ be a homomorphism from some CQ $p$ in $q$ to $\chase{\Omc}{\mn{ext}(D, \uparrow, D_0)}$ and $A=\mn{ran}(h) \cap \adom{D}$. Then, there exists an $\Sbf$-database $D'$ and a homomorphism $\downarrow$ from $D'$ to $D$ such that: 
	\begin{enumerate}
		\item $\downarrow$ is the identity function on $\mn{ker}(D')$;
	    \item $(D', \downarrow)$ is a diversification of $D$;
		\item $\mn{ker}(D') \subseteq \mn{ran}(h) \cap \adom{D}$;
		\item $\mn{ext}(D', \uparrow \circ \downarrow, D_0) \models Q$.
	\end{enumerate}
	\end{lemma}

\begin{proof}

We start by partitioning the set of variables from the CQ $p$, $\adom{p}$, into sets $A_0, A_1, \dots A_n$, such that: 
\begin{enumerate}
	\item $h(A_0) \subseteq \adom{D}$; 
	\item for every $i>0$, $h(A_i) \subseteq \adom{D^{\abf_i}_0} \setminus \adom{D}$, \item for every $i, j >0$, $i \neq j$ implies $\abf_i \neq \abf_j$.
\end{enumerate}

Thus, $A_0$ contains all variables which map into $\adom{D}$, and each $A_i$ contains a maximal set of variables which map into constants from some guarded unraveling of $D_0$ which are not from $D$.  

Further on, for every $i$, with $0 \leq 1 \leq n$, we define a set of constants $\bbf_i$ such that $\bbf_i$ is obtained from $\abf_i$ by replacing every constant which is not from $\mn{ran}(h)$ (i.e. it is not `hit' by $h$) with a fresh constant.

We now proceed to definining $D'$. We start by initializing $D'$ as $\emptyset$. For every atom $r(\cbf)$ in $D$ for which $\cbf \cap \mn{ran}(h) \neq \emptyset$, we rename all constants from $\cbf$ which are not from $\mn{ran}(h)$ as fresh constants and add the atom to $D'$. Then, for every $i$, with $1 \leq i \leq n$, there must be some atom $r(\cbf_i)$ in $D$ such that $\abf_i \subseteq \cbf_i$.  We construct a new tuple of constants $\cbf'_i$ from $\cbf_i$ as follows: constants from $\abf_i$ are replaced with constants from $\bbf_i$ and constants which are not from $\abf_i$ are freshly renamed. We add $r(\cbf'_i)$ to $D'$. 

By construction, $\adom{D} \cap \mn{ran}(h) \subseteq \adom{D'}$. All other constants from $\adom{D'}$ are fresh: this is due to the fact that all facts of the form $r(\cbf'_i)$ added to $D'$ at the step above are distinct; thus, when constructing $D'$ we do not re-use fresh constants introduced when defining the sets $b_i$ twice. Thus,  $\mn{ker}(D') \subseteq \adom{D} \cap \mn{ran}(h)$ (Point (3) of the Lemma). We define $\downarrow$ as a mapping from $\adom{D'}$ to $\adom{D}$ which is the identity on $\mn{ran}(h) \cap \adom{D}$ and which maps fresh constants to original constants for the remaining elements of $\adom{D'}$. It can be verified easily that Point (1) and Point (2) of the Lemma are met. It remains to show Point (4).

Let $g=\uparrow \circ \downarrow$ and $D^+=\mn{ext}(D', g, D_0)$. 
We construct a mapping $h'$ from $\adom{p}$ to $\adom{D^+}$ as follows: $h'(x)=\begin{cases} h(x) &\mbox{if } x \in A_0 \\
	\kappa_i(h(x)) & \mbox{if } x \in A_i \end{cases}$.

For every $i$ with $0 <i \leq n$, we define an isomorphism $\kappa_i$ from the database $D_0^{\abf_i}$, the guarded unraveling of $D_0$ added to the guarded set $\abf_i$ during the construction of $\mn{ext}(D, \uparrow, D_0)$ to the database $D_0^{\bbf_i}$, the copy of $D_0^{\abf_i}$ added to $D'$ during the construction of $D^+$, in which constants from $\abf_i$ have been replaced with constants from $\bbf_i$. Due to the property of guarded unravelings to preserve atomic consequences, it can be shown that for every $i$ with $0 <i \leq n$, $\kappa_i$ is an isomorphism also from $\chase{\Omc}{\mn{ext}(D, \uparrow, D_0)}_{|\adom{D_0^{\abf_i}}}$ to $\chase{\Omc}{D^+}_{|\adom{D_0^{\bbf_i}}}$.

Further on, due to the fact that $(D', \downarrow)$ is a diversification of $D$ and again to the property of guarded unravelings, it can be shown that for every guarded set $\bbf$ in $D'$, $\downarrow_{\bbf}$ is an isomorphism from $\chase{\Omc}{D^+}_{|\adom{D_0^{\bbf}}}$ to $\chase{\Omc}{\mn{ext}(D, \uparrow, D_0)}_{|\adom{D_0^{\downarrow(\bbf)}}}$, and thus for guarded sets $\abf$ in $D'$ such that $\abf \in \adom{D}$, $\chase{\Omc}{D^+}_{|\adom{D_0^{\abf}}}$ and $\chase{\Omc}{\mn{ext}(D, \uparrow, D_0)}_{|\adom{D_0^{\downarrow(\bbf)}}}$ are identical. By putting all this together it can be verified that $h'$ is a homomorphism from $p$ to $\chase{\Omc}{D^+}$. 
\qed

\end{proof}

\subsection{Characteristic Databases and Covers}
\label{subsec:charac}

In this section we put together the notions introduced in previous subsections to define \emph{(extended) characteristic databases} and \emph{covers} of an OMQ.

\begin{definition}
	\label{def:characteristic}
	For $Q$ an OMQ from (GDLog, UCQ),  \emph{the set of characteristic databases for $Q$} is  $\class{D}_Q=\{D \mid (D, \uparrow) \in \mn{mdiv}(D_0, Q), D_0 \mbox{ is qi w.r.t. } Q\}$. \emph{The set of extended characteristic databases for $Q$} is $\class{D}^+_Q=\{\mn{ext}(D, \uparrow, D_0) \mid (D, \uparrow) \in \mn{mdiv}(D_0, Q), D_0 \mbox{ is qi w.r.t. } Q\}$.
	\end{definition}


	

Thus, a characteristic database $D$ is the database part of a minimal diversification of some database $D_0$ which is qi w.r.t. $Q$, while an extended characteristic database is the extension of such a database $D$ with guarded unravelings of $D_0$ according to the homomorphism from the diversification. Due to the minimality condition on diversifications, characteristic databases satisfy the following properties:

\begin{lemma}
	\label{lem:mindivers-targets}
	Let $Q=(\Omc, \Sbf, q)$ be an OMQ from $(\text{GDLog}, \text{UCQ})$,  $D^+$ be an $\Sbf$-database of the form $\mn{ext}(D, \uparrow, D_0)$ from $\class{D}_Q^+$, and $h$ a homomorphism from a CQ $p$ in $q$ to $\chase{\Omc}{D^+}$. Then the following hold: 
	\begin{enumerate}
		\item  $\mn{ker}(D) \subseteq \mn{ran}(h)$;
		\item  there exists a computable function $f$ such that $|D| \leq f(|Q|)$.	\end{enumerate}

\end{lemma}

\begin{proof}
	We start by showing Point (1). From Point (2) of Lemma~\ref{lem:divers-constr} we know that there exists a database $D'$ and a homomorphism $\downarrow$ from $D'$ to $D$ such that $(D', \downarrow)$ is a diversification of $D$. Thus, $D' \preceq D$. As $(D, \uparrow)$ is a diversification of $D_0$, it follows that $(D', \uparrow \circ \downarrow)$ is a diversification of $D_0$. From Point (4) of Lemma~\ref{lem:divers-constr} we obtain that $(D', \uparrow \circ \downarrow)$ is a diversification of $D_0$ w.r.t. $Q$. In other words, $D' \in \mn{div}(D_0, Q)$. As $D \in \mn{mdiv}(D_0, Q)$, and $D' \preceq D$, it must be the case that $D \preceq D'$ (otherwise, $D$ would not be minimal). Thus, $D \leftrightarrow D'$. 
	
	From Point (1) of Lemma~\ref{lem:divers-constr}, we know that $D'$ maps into $D$ via the function $\downarrow$  which is the identity on $\mn{ker}(D')$. Thus, $\mn{ker}(D') \subseteq \mn{ker}(D)$ (a non-isolated constant, i.e. a constant which occurs in $\mn{ker}(D')$ can only map into another non-isolated constant, thus a constant from $\mn{ker}(D)$). As $D$ is a core, it is the case that $D \rightarrow^{io} D'$, otherwise the composition of a non-injective homomorphism to $D'$ with the homomorphism $\downarrow$ from $D'$ to $D$ would lead to a non-injective endomorphism. As every homomorphism from $D$ to $D'$ maps $\mn{ker}(D)$ into $\mn{ker}(D')$ injectively, it is the case that $|\mn{ker}(D)| \leq |\mn{ker}(D')|$. Thus, $\mn{ker}(D')= \mn{ker}(D)$. Finally, from Point (3) of Lemma~\ref{lem:divers-constr} we know that $\mn{ker(D')} \subseteq \mn{ran}(h) \cap \adom{D}$, and thus $\mn{ker}(D) \subseteq \mn{ran}(h)$. 
	
	\medskip 
	
	Point (2) follows from Point (1) and from the fact that $D$ is a core: as $|\mn{ker}(D)|$ is bounded in $|Q|$, $|D|$ will be bounded in $|Q|$ as well. 	
	\qed
\end{proof}	
\medskip

We next define the \emph{cover} $Q_c=(\Omc, \Sbf, q_c)$ of an OMQ $Q=(\Omc, \Sbf, q)$ by using the set of extended characteristic databases for $Q$, $\class{D}^+_Q$. To construct CQs in $q_c$, we will consider images of CQs from $q$ in the chase of some database $D^+$ from $\class{D}^+_Q$ together with extra atoms (facts)\footnote{As we construct a CQ by conjoining the image of an original CQ into a database with more facts from the database, the distinction between facts and atoms is blurred.} from $D^+$ which will guard atoms in the CQ image. We want to include guards from $D^+$ which cover as many atoms as possible from the image of the CQ, and thus to potentially decrease submodular width as in Example~\ref{ex:cover}. Formally: 

\begin{definition}
	\label{def:cover}
	 For $Q=(\Omc, \Sbf, q)$ an OMQ from (GDLog, UCQ), its \emph{cover} is the OMQ  $Q_c=(\Omc, \Sbf, q_c)$, with $q_c$ the union of all CQs $p_c$ with $D[p_c]$ of the form $h(p) \cup S$, where:
	\begin{enumerate}
		\item  $h$ is a homomorphism from some CQ $p$ in $q_c$ to $\chase{\Omc}{D^+}$,  where $D^+$ is a database of the form $\mn{ext}(D, \uparrow, D_0)$ from  $\class{D}^+_Q$;
		\item $S$ is a minimal set of atoms such that: 
		\begin{enumerate}
			\item $D \subseteq S \subseteq D^+$;
			\item for every atom $r(\abf)$ from $h(p)$, there exists an atom $r'(\abf')$ from $S$ such that $\abf \subseteq \abf'$ and $\abf'$ is a maximal guarded set in $D^+$.
		\end{enumerate}
\end{enumerate}
\end{definition} 
 
The set of atoms $S$ in Definition~\ref{def:cover} is the set of guards added to the image $h(p)$ of a CQ $p$ in $q$. By including $D$, $S$ trivially guards every maximal set of atoms $S'$ from $h(p)$ for which $\adom{S'}$ is a guarded set in $D$. As concerns remaining atoms from $h(p)$ (those which map into the guarded unraveling part of $D^+$), we note that for every maximal set $S'$ of such atoms for which $\adom{S'}$ is a guarded set in $D^+$, there exists a unique atom $r'(\abf')$ in $D^+$ such that $\adom{S'} \subseteq \abf'$. Thus, implicitly, the definition insures that `maximal' guards, i.e. guards which cover as many atoms, are added to $h(p)$.

\begin{example}  Let $\class{Q}$ be as in Example~\ref{ex:cover}. For every $i>1$, let $D_{i,0}$ be the $\Sbf_i$-database: $\{S_i(\abf_i), T(a_1, a_2), T(a_1, a_3), \dots T(a_{i-1}, a_i)\}$.
	We have that $D_{i,0} \models Q_i$ and furthermore, $D_{i,0}$ is qi w.r.t. $Q_i$:  $q_i$ is the only contraction mapping into $\chase{\Omc_i}{D_{i,0}}$. Let $D_i$ be the $\Sbf_i$-database $\{S(\abf_i)\}$. Then $(D_i, \uparrow) \in  \mn{mdiv}(D_{i,0}, Q)$, where $\uparrow$ is the identity function on $\abf_i$: any guarded unraveling of $D_{i,0}$ at $\abf_i$ will add back all the facts from $D_{i,0}$.  Thus, $D_i \in \class{D}_{Q_i}$ and for every database of the form $ \mn{ext}(D_i, \uparrow, D_{i,0})$, $D_{i,0} \subseteq \mn{ext}(D_i, \uparrow, D_{i,0})$ (depending on guarded unravelings, $\mn{ext}(D_i, \uparrow, D_{i,0})$ might contain some superfluous facts). In the following we fix such a database $\mn{ext}(D_i, \uparrow, D_{i,0})$ and refer to it as $D_i^+$.
	
	 The CQ $q_i$ maps then into $\chase{\Omc_i}{D_i^+}$ via a homomorphism $h$ with range $\abf_i$. We construct a CQ $p_c$ belonging to the cover of $Q_i$ based on $h$ and $D_i^+$: $D[p_c]$ will be the union of $\psi_i^R(\abf_i) \wedge \psi_i^T(\abf_i)$, the image of $q_i$ under $h$, with the singleton set $\{S(\abf_i)\}$ ($D_i$) from $D_i^+$. No extra atoms are needed to be added ($S=D$) as $S(\abf_i)$ is a guard for all atoms from $h(q_i)$. Thus $D[p_c]= \psi_i^R(\abf_i) \wedge \psi_i^T(\abf_i) \wedge  S(\abf_i)$, i.e it is isomorphic to $D[q'_i]$ from Example~\ref{ex:cover}. It can be shown that it is the unique such CQ (up to isomorphism) and thus the cover of $Q_i$ is identical (up to isomorphism) to the OMQ $Q'_i$ from Example~\ref{ex:cover}.		
\end{example}

As next lemma shows, covers are equivalent to original OMQs. 

\begin{lemma}
	\label{lem:coverProperties}
	For $Q=(\Omc, \Sbf, q)$ an OMQ from (GDLog, UCQ), and $Q_c$ its cover, $Q \equiv Q_c$. 
\end{lemma}

\begin{proof}
		
It is clear that $Q_c \subseteq Q$. We show that $Q \subseteq Q_c$. For $D$ an $\Sbf$-database such that $D \models Q$, there exists a database $D_0 \to D$ which is qi w.r.t. $Q$. Furthermore, there exists a diversification $(D_1, \uparrow) \in \mn{mdiv}(D_0, Q)$. Thus, $\mn{ext}(D_1, \uparrow, D_0) \in \class{D}^{+}_{Q}$, and there exists some CQ $p$ in $q$ s.t. $p$ maps into  $\chase{\Omc}{\mn{ext}(D_1, \uparrow, D_0)}$ via some homomorphism $h$.  We construct a CQ $p_c$ from $q_c$ as in Definition~\ref{def:cover} based on $p$, $h$, and $\chase{\Omc}{\mn{ext}(D_1, \uparrow, D_0)}$. Then, $D[p_c] \subseteq \chase{\Omc}{\mn{ext}(D_1, \uparrow, D_0)}$, thus $p_c \to  \chase{\Omc}{\mn{ext}(D_1, \uparrow, D_0)}$, and as $p_c$ is a CQ in $q_c$, it follows that $\mn{ext}(D_1, \uparrow, D_0) \models Q_c$. As $\mn{ext}(D_1, \uparrow, D_0) \to D_0 \to D$, it follows that $D \models Q_c$.

\qed
\end{proof}

Thus covers are a good candidate for witnesses of low semantic SMW for OMQs from $(\text{GDLog}, \text{UCQ})$. We defer to the next section the result detailing in which way they serve as such witnesses. We  next show that they have finite bounded size:

\begin{lemma}
	\label{lem:coversize}
	For  $Q=(\Omc, \Sbf, q)$ an OMQ from  $(\text{GDLog}, \text{UCQ})$ and $Q_c=(\Omc, \Sbf, q_c)$ its cover, there exists a computable function $g$ such that for every CQ $p_c$ in $q_c$, $|p_c| \leq g(|Q|)$.
\end{lemma}

\begin{proof}
Each $p_c$ has the form $h(p) \cup S$, where $p$ is some CQ in $q$, $h$	a homomorphism from $p$ to $\chase{\Omc}{\mn{ext}(D, \uparrow, D_0)}$, for some $\mn{ext}(D, \uparrow, D_0) \in \class{D}^+(Q)$ and $S$ the extension of $D$ with atoms from $\mn{ext}(D, \uparrow, D_0) \in \class{D}^+(Q)$ which guard atoms in $h(p)$. As $|h(p)|$ is bounded in $|Q|$, $|D| \leq f(|Q|)$, for some computable function $f$ (from Lemma~\ref{lem:mindivers-targets}) and $S$ contains at most one atom for each atom in $h(p)$, it follows that $|p_c|$ is bounded in $|Q|$. 

	\end{proof}

Finally, we show that:

\begin{lemma}
	\label{lem:coverComput}
	
	For every OMQ $Q=(\Omc, \Sbf, q) \in (\text{GDLog}, \text{UCQ})$, $Q_c=(\Omc, \Sbf, q_c)$ is computable.
	\end{lemma}

\begin{proof}
	
We denote with $\Sbf'$ the schema $\Sbf$ extended with relational symbols which occur in $Q$ and we define \emph{the diameter of $Q$} to be the maximum between the maximum arity of some relational symbol from $\Sbf$ and the number of variables in some CQ in $q$. We will use Lemma~\ref{lem:coversize} in conjunction with an encoding into Guarded Second Order Logic (GSO)  \cite{GrHO02} to check whether some structure over the extended signature $\Sbf'$ of size bounded by $g(|Q|)$  is of the right form to be part of $q_c$. While GSO is in general undecidable, on sparse structures it has the same expressivity as MSO \cite{Blum10}.  For our purposes, it is possible to restrict to structures of treewidth bounded by the diameter of the given OMQ, which are sparse \cite{Cour03}. This is due to the following property of OMQs from (GDLog, UCQ) \cite{Feier-Lutz-Kuusisto-LMCS-19}:

\begin{nclaim}
	Let $Q$ be an OMQ from (GDLog, UCQ) and let $d$ be the diameter of $Q$. Then, for every database $D$ such that $D \models Q$, there exists a database $D'$ such that $D' \to D$, $D' \models Q$ and $D'$ has treewidth at most $d$. 
\end{nclaim}

One potential difficulty for the GSO encoding is the usage of guarded unravelings in our definitions. To overcome this, we provide an alternative characterization for CQs from $q_c$ which does not use  guarded unravelings.

 For a CQ $p$, an \emph{adornment} of $p$ is a tuple $(\delta, S)$ where $\delta=(T, \chi)$, with $T=(V, E)$  a tree with root $v_0$, is a tree decomposition of $D[p]$ and $S$ is a minimal $\Sbf$-database such that: 
\begin{enumerate}
	\item for every atom $r(\abf)$ in $p$, $\abf$ is guarded in $S$;
	\item for every $v \in V \setminus \{v_0\}$, $\chi(v)$ is guarded in $S$.
\end{enumerate}
Intuitively, an adornment extends a CQ with guards from an $\Sbf$-database and guesses which part of the CQ maps into a main database (the atoms from $p_{|\chi(v_0)}$), and which part maps into guarded unravelings attached to the main database (the remaining atoms from $p$). 

An \emph{extended adornment} of $p$ is a tuple $(\delta, S, \uparrow, D_0)$ such that $(\delta, S)$ is an adornment for $p$ as above, $\uparrow$ is a homomorphism from $S$ to $D_0$ which is i.g.s. and which furthermore has the property that $(D, \uparrow)$ is a diversification of $D_0$, where $D=S_{|\chi(v_0)}$. An extended adornment of $p$ is \emph{valid w.r.t. $Q$} if for every atom $r(\abf)$ in $p$, $D_0 \models (\Omc, \Sbf, r(\uparrow(\abf)))$. Intuitively, an extended adornment of $p$ which is valid w.r.t. $Q$ is an alternative representation for a database $D^+$ of the form  $\mn{ext}(D, \uparrow, D_0)$ for which $p \to \chase{\Omc}{D^+}$, together with some aditional information regarding how $p$ maps into $\chase{\Omc}{D^+}$. Note that guarded unravelings do not figure explicitly in the definition of an extended adornment which is valid w.r.t. $Q$. However, due to the property of guarded unravelings regarding preservation of answers to atomic queries, the requirement that for every atom $r(\abf)$ in $p$ it holds that $D_0 \models (\Omc, \Sbf, r(\uparrow(\abf)))$, is nothing else but requiring that 
$D_0^{\uparrow(\abf)} \models (\Omc, \Sbf, r(\uparrow(\abf))) $, where $D_0^{\uparrow(\abf)}$ is the guarded unraveling of $D_0$ at $\uparrow(\abf)$. The following claim captures the connection between valid adornments and databases of the form $\mn{ext}(D, \uparrow, D_0)$ which entail $Q$:

\begin{nclaim}
	\label{claim:ext_adorn}
	For $p$ some contraction of a CQ in $q$ and  $(\delta, S, \uparrow, D_0)$ an extended adornment of $p$ which is valid w.r.t. $Q$, with $\delta=(T, \chi)$, and $T=(V, E)$ a tree with root $v_0$, it is the case that $(S_{|\chi(v_0)}, \uparrow) \in \mn{div}(D_0, Q)$. 
	\end{nclaim}

\begin{proof}
	Let $D=S_{|\chi(v_0)}$ and $D^+=\mn{ext}(D, \uparrow, D_0)$. It is possible to construct a homomorphism $h'$ from $p$ to $\chase{\Omc}{D^+}$, by taking $h'(x)=h(x)$, if $h(x) \in \chi(v_0)$ and then inductively traverse the tree decomposition $\delta$ and for each $v \in V \setminus \{v\}$, define $h'$ on  $\chi(v)$  guided by $h$ and $\uparrow$. 
	\end{proof}

An extended adornment $(\delta, S, \uparrow, D_0)$ of $p$ which is valid w.r.t. $Q$ is \emph{minimal} if there exists no extended adornment $(\delta', S', \downarrow, D_0)$ of some contraction $p'$ which is valid w.r.t. $Q$, with $\delta'=(T', \chi')$ a tree decomposition of $p'$ with root $v'_0$, such that $D' \preceq D$ and $D \not \preceq D'$, where $D'=S'_{|\chi'(v'_0)}$ and $D=S_{|\chi(v_0)}$.

Then, the following hold: 

\begin{nclaim}
	\label{claim:alt_charac}
	A CQ $p_c$ is part of $q_c$ iff $D[p_c]$ is isomorphic to a database of the form $D[p'] \cup S$ with $p'$ the contraction of some CQ in $q$ for which there exists an extended adornment $(\delta, S, \uparrow, D_0)$ of $p$ which is valid w.r.t. $Q$ and minimal, with $D_0$ an $\Sbf$-database which is qi w.r.t. $Q$. \end{nclaim}
\begin{proof}
	
$`\Rightarrow'$: Assume $p_c$ is a CQ in $q_c$. Then, $p_c$ is of the form $h(p) \cup S$, where $p$ is some CQ in $q$, $h$ is a homomorphism from $p$ to a database $\mn{ext}(D, \uparrow, D_0)$ from $\class{D}^+_Q$, and $S$ is a minimal set such that $D \subseteq S \subseteq D^+$ and for every atom $r'(\abf')$ in $h(p)$, there exists an atom $r(\abf)$ in $S$ with $\abf' \subseteq \abf$. We construct a tree decomposition $\delta=(T, \chi)$ of $h(p)$ such that $\chi(v_0)=\adom{D}$, and for every $v \neq v_0$ $\chi(v)$ is maximal such that it is guarded in $S \setminus D$. Then, $(\delta, S)$ is an adornment of $h(p)$. It can also be easily seen that $(\delta, S, \uparrow, D_0)$ is an extended adornment of $h(p)$ which is valid w.r.t. $Q$: this is due to the fact that $(D, \uparrow) \in \mn{div}(D_0)$. To see that it is minimal, assume that this is not the case. Then, there exists another extended adornment $(\delta', S', \uparrow', D_0)$ of some contraction $p'$ of a CQ from $q$ which is valid w.r.t. $Q$ and $D' \preceq D$ and $D \not \preceq D'$ where $D'=S'_{|\chi'(v'_0)}$. But then, from Claim~\ref{claim:ext_adorn} we know that $(D', \uparrow') \in \mn{div}(D_0, Q)$. Thus, $(D=S_{|\chi(v_0)}, \uparrow)\notin \mn{mdiv}(D_0, Q)$ -- contradiction.

$`\Leftarrow'$: Assume that there exists an extended adornment $(\delta, S, \uparrow, D_0)$ of some contraction $p$ which is valid w.r.t. $Q$ and minimal and such that $D_0$ has the qi property w.r.t. $Q$. Let $D=S_{|\chi(v_0)}$. From Claim~\ref{claim:ext_adorn} it follows that $(D, \uparrow) \in \mn{div}(D_0, Q)$. Assume that $(D, \uparrow) \notin \mn{mdiv}(D_0, Q)$. Then, there exists  $(D', \uparrow') \in \mn{mdiv}(D_0, Q)$ such that $D' \preceq D$ and $D \not \preceq D'$. Then, there must be some CQ $p'$ which maps via a homomorphism $h'$ to $\chase{\Omc}{\mn{ext}(D', \uparrow', D_0)}$. Then, it is possible to construct an extended adornment $(\delta', S', \uparrow'', D_0)$ of $h'(p')$ which is valid w.r.t. $Q$ such that $D'=S'_{|\chi(v_0)}$ -- contradiction with the fact that $(\delta, S, \uparrow, D_0)$ is minimal. Thus, $(D, \uparrow) \in \mn{mdiv}(D_0, Q)$. Then, $p \cup S$ fulfills all conditions from Definition~\ref{def:cover} and thus it is isomorphic to a CQ in $q_c$.

	\end{proof}

Thus, to compute $q_c$, it is enough to iterate through tuples $(p, S)$ such that $p$ is a contraction of a CQ in $q$, for which it holds that $|D[p] \cup S|$ is bounded by $g(|Q|)$, and to check whether there exist an adornment of $p$ of the form $(\delta, S)$ which can be extended to an extended adorment $(\delta, S, \uparrow, D_0)$ of $p$ which is valid w.r.t. $Q$ and minimal and for which $D_0$ is qi w.r.t $Q$. To perform such checks, we need formulas of the following type: 

\begin{itemize}
	\item $\phi_{\mn{qi}}(D_0)$: evaluates to true when $D_0$ is an $\Sbf$-database which has the qi property w.r.t. $Q$;
	\item $\phi_{\mn{div}}(D, \uparrow, D_0)$: evaluates to true if $(D, \uparrow)$ is a diversification of $D_0$;
\item $\phi_{\prec}(D_1, D_2)$: evaluates to true if $D_1 \preceq D_2$ and $D_2 \not \preceq D_1$;
\item $\phi_{\mn{c}}(p)$: evaluates to true if $p$ is a contraction of some CQ in $q$;
\item $\phi_{\mn{ad}}(p, D, S)$: evaluates to true if there exists an adornment $(\delta, S)$ of $p$ with $\delta=(\chi, T)$, and $T=(V, E)$ a tree with root $v_0$, such that $S_{|\chi(v_0)}=D$;
	\item $\phi_{\mn{eval}}(p, \uparrow, D)$: for $p$ a CQ and $\uparrow$ a function from $\adom{p}$ to $\adom{D}$ which is i.g.s., it evaluates to true when for every atom $r(\abf)$ from $p$, it holds that $D \models (\Omc, \Sbf, r(\uparrow(\abf)))$;
	\item $\phi_{\mn{va}}(p, D, S, \uparrow, D_0)$: evaluates to true if 
    $\uparrow$ is a  homomorphism from $S$ to $D_0$ which is i.g.s. and $\phi_{\mn{div}}(D, \uparrow, D_0)$ holds and $\phi_{\mn{eval}}(p, \uparrow, D)$ holds.
\end{itemize}

We do not provide the actual encodings for the above formulas, as while they are long and technical, they are also rather straightforward. For examples of similar encodings, the reader can consult the proof of Lemma~D.5 in \cite{BDFLP-PODS20}. Given such formulas, we construct a formula:

\begin{align*}
	\Phi(p,S)&=\exists D, \uparrow, D_0 \ \phi_{\mn{ad}}(p,D, S) \wedge  \phi_{\mn{qi}}(D_0) \wedge \phi_{\mn{va}}(p, D, S, \uparrow, D_0) \wedge \forall p', D', S', \uparrow'
	\\
	& \quad \neg (\phi_{\mn{c}}(p') \wedge \phi_{\mn{ad}}(p', D', S') \wedge \phi_{\mn{va}}(p', D', S', \uparrow, D_0) \wedge \phi_{\prec}(D', D))
\end{align*}

Then $D[p] \cup S \models \Phi(p, S)$ iff $D[p] \cup S$ is isomorphic to $D[p_c]$, with $p_c$ from $Q_c$.



%
%
%
%

\end{proof}

\section{Main Results}
\label{sec:results}

In this section we establish the main results of our paper: a syntactic characterization for fixed-parameter tractability of OMQs from $(\text{GDLog}, \text{UCQ})$, a reduction from parameterized uniform CSPs to parameterized OMQs from $(\text{GDLog}, \text{UCQ})$, and a semantic characterization for fixed-parameter tractability of OMQs from $(\text{GTGD}, \text{UCQ})$. 

\subsection{Results for (GDLog, UCQ)}

For a class $\class{Q}$ of OMQs from  $(\text{GDLog}, \text{UCQ})$, we denote with $\class{D}_{\class{Q}}$ the class of characteristic databases for OMQs from $\class{Q}$ and with $\class{Q}_{c}$ the class of covers for OMQs from $\class{Q}$. The first main result which we show in this section is as follows:

\begin{theorem}[Main Result 3]
	\label{thm:unbArityGDLog}
	Let $\class{Q}$ be a r.e. enumerable class of OMQs from $(\text{GDLog}, \text{UCQ})$. Under the Exponential Time Hypothesis, the following are equivalent: 
	
	\begin{enumerate}
		\item $\text{p-OMQ}(\class{Q})$ is fixed-parameter tractable
		\item $\class{Q}_{c}$ has bounded sub-modular width 
		\item $\class{D}_{\class{Q}}$ has bounded sub-modular width. 
	\end{enumerate}
	
\end{theorem}

Towards showing the result, the following fpt-reduction from evaluation of parameterized CSP to evaluation of parameterized OMQs from  $(\text{GDLog}, \text{UCQ})$ plays an important role. 


\begin{theorem}[Main Result 2]
		\label{thm:pCSPtopOMQ}
Let $Q=(\Omc, \Sbf, q)$ be an OMQ from (GDLog, UCQ). Then, there exists an fpt-reduction from $\text{p-CSP}(\class{D}_{Q},\_)$ to $\text{p-OMQ}(\{Q\})$. 
\end{theorem}


\begin{proof}
	Let $(D, B)$ be an instance of $\text{p-CSP}(\class{D}_{Q},\_)$. We construct an \Sbf-database $D_2^+$ and show that $D  \to B$ iff $D_2^+ \models Q$, as follows: let $\pi$ be  the projection mapping from $D \times B$ to $D$ and let $D_2$ be the database obtained from $D \times B$ by dropping all atoms $R(\abf)$ for which $\pi|_{\abf}$ is not injective. Then, $\pi|_{\adom{D_2}}$ is a homomorphism from $D_2$ to $D$ which is i.g.s.  As $D \in \class{D}_Q$, there exists a database $D_0$ which is qi w.r.t. $Q$ and some diversification $(D, \uparrow) \in \mn{mdiv}(D_0, Q)$ such that $\mn{ext}(D, \uparrow, D_0) \in \class{D}^+_Q$. We will refer in the following to $\mn{ext}(D, \uparrow, D_0)$ as $D^+$. Then  $\pi \circ \uparrow$ is a homomorphism from $D_2$ to $D_0$ which is i.g.s. We define $D_2^+$ as $\mn{ext}(D_2, \pi \circ \uparrow, D_0)$. 
						
	`$\Rightarrow$': Assume that $D_2^+ \models Q$. The strategy of the proof is as follows: using Lemma~\ref{lem:divers-constr} we construct a diversification $(D_3, \iota)$ of $D_2$ and then using the qi property of $D_0$, and subsequently of $D^+$, we show that  $D_3 \preceq D$ and that  $(D_3, g) \in \mn{div}(D_0, Q)$, for some homomorphism $g$. But $(D, \uparrow)$ is a minimal diversification of $D_0$ w.r.t. $Q$. Thus, $D \preceq D_3$, so $D \rightarrow D_3 \rightarrow D_2 \rightarrow D \times B \rightarrow B$.
	
	Now to the details. Let $p'$ be some contraction of a CQ $p$ in $Q$ which maps injectively only via an (injective) homomorphism $h'$ into $\chase{\Omc}{D_2^+}$ and let $A=\mn{ran}(h') \cap \adom{D_2}$.  Also, let $\pi^+$ be the extension of the projection homomorphism $\pi$ from $D_2$ to $D$ to a homomorphism from $D_2^+$ to $D^+$. Then, $\pi^+$ is a homomorphism also from $\chase{\Omc}{D_2^+}$ to $\chase{\Omc}{D^+}$ and $\pi^+ \circ h'$ is a homomorphism from $p'$ to $\chase{\Omc}{D^+}$. As $D_0$ is qi w.r.t $Q$, according to Lemma~\ref{lem:io-prop}, so is $D^+$. From Definition~\ref{def:qi-initial}, as $p' \rightarrow^{io} \chase{\Omc}{D_2^+}$, and $D_2^+ \rightarrow D^+$, it must be the case that $p' \rightarrow^{io} \chase{\Omc}{D^+}$, and thus $\pi^+ \circ h'$ is injective. Further on, as $A \subseteq \mn{ran}(h')$, it follows that  $\pi^+_{|A}$ is injective. Also, as $A \subseteq \adom{D_2}$, it follows that $\pi_{|A}$ is injective.
	
	As $D_2^+$ is of the form $\mn{ext}(D_2, \pi \circ \uparrow, D_0)$, according to  Lemma~\ref{lem:divers-constr}  there exists a diversification $(D_3, \iota)$ of $D_2$ such that $\iota$ is the identity on $\mn{ker}(D_3)$, $\mn{ker(D_3)} \subseteq A$, and $\mn{ext}(D_3, g, D_0) \models Q$, where $g$ is the composition homomorphism $\uparrow \circ \xspace \pi \xspace \circ \iota$ from $D_3$ to $D_0$. Let $D_3^+$ be $\mn{ext}(D_3, g, D_0)$. Note that $g$ is i.g.s. Also, let $h=\pi \xspace \circ \iota$ be the composition homomorphism from $D_3$ to $D$. It can be seen that $h$ is also i.g.s. Figure~\ref{fig:exp} tries to provide an overview of all the homomorphisms employed in the proof. Note that while the figure depicts extended databases, e.g. databases like $\mn{ext}(D, \uparrow, D_0)$ ($D^+$), most of the depicted homomorphisms are between the `base' databases, e.g. $D$, $D_3$, etc.

	\begin{figure}
	\begin{boxedminipage}{\columnwidth}
	
		\begin{tikzpicture}[auto, scale=0.8]
		\GraphInit
		
		\node(n1) at (0.5, 0){$p$};
		\node(n2) at (2.3, 0) {$p'$};
		\node(n3) at (5.5, 0) {$\chase{\Omc}{D_2^+}$};
		\node(n4) at (10.5, 0) {$\chase{\Omc}{D^+}$};
		\node(n5) at (14.5, 0) {$\chase{\Omc}{D_0}$};
		\node(n6) at (2, -1.5) {$\chase{\Omc}{D_3^+}$};
		\node(n7) at (4.3, 0.2){\scriptsize{io}};
		\node(n8) at (9.35, 0.3){\scriptsize{io}};
		\draw [->]  (n1)  to (n2) ;
				\draw [->]  (n1)  to (n6) ;
		\draw [->] (n2)  to  node [midway, above] {\small{$h'$}} (n3);
		
		\draw [->] (n3)  to  node [midway, above] {\small{$\pi$}} node [midway, below] {\small{i.g.s.}} (n4);
		\draw [->] (n4)  to node [midway, above] {\small{$\uparrow$}} node [midway, below] {\small{i.g.s.}} (n5); 
		\draw (n6) [->, out=0, in=-170]   to node [midway, above] {\small{$\iota$}} node [midway, below] {\small{i.g.s.}} (6, -0.3); 
		\draw (n2) [->, out=25, in=155]  to (n4) ;
		\draw (n6) [->, out=-5, in=-170] to node [midway, above] {\small{$h$}} node [midway, below] {\small{i.g.s.}}(10.8, -0.3) ;
		\draw (n6) [->, out=-10, in=-170] to node [midway, above] {\small{$g$}} node [midway, below] {\small{i.g.s.}} (14.8, -0.3) ;
				
		\end{tikzpicture}
\end{boxedminipage}
\label{fig:exp}
	\caption{Constructions for Proof of Direction $\Rightarrow$ of Theorem~\ref{thm:pCSPtopOMQ}}
	\end{figure}
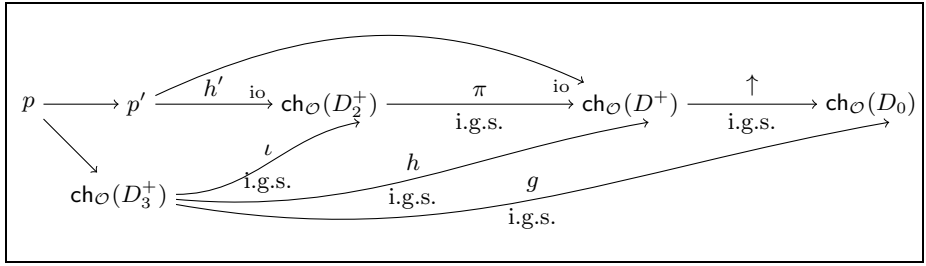
	
	We show next that $h|_{\mn{ker}(D_3)}$ is also injective. We have that $\mn{ran}(\iota_{|\mn{ker(D_3)}})=\mn{ker(D_3)}$, thus $h|_{\mn{ker}(D_3)}= \pi \xspace \circ \iota_{|\mn{ker(D_3)}}$ is the same as $ \pi_{|\mn{ker(D_3)}} \xspace \circ \iota_{|\mn{ker(D_3)}}$. As $\mn{ker(D_3)} \subseteq A$, $\pi_{|A}$ is injective, and $\iota_{|\mn{ker(D_3)}}$ is the identity function, it follows that $h|_{\mn{ker}(D_3)}$ is injective. Recall that $h$ is also i.g.s. Thus, $(D_3, h)$ is a diversification of $D$ and $D_3 \preceq D$. As $D \preceq D_0$, it follows that $D_3 \preceq D_0$. At the same time, $\mn{ext}(D_3, g, D_0) \models Q$, thus $(D_3, g) \in \mn{div}{D_0, Q}$. But $(D, \uparrow) \in \mn{mdiv}(D_0, Q)$, thus $D \preceq D_3$. Thus, $D \to D_3 \to D_2 \to D \times B \to B$.
	
	\smallskip
	
	`$\Leftarrow$': Assume that $D \to B$. Then, $D$ maps into $D \times B$ via some homomorphism $h$. Let $A=\mn{ran}(h)$. At the same time, $D \times B$ maps into $D$ via the projection mapping $\pi$. Thus, $\pi \circ h$ is a homomorphism from $D$ to itself. As $D$ is a core, $\pi \circ h$ must be injective, and thus $\pi|_A$ is injective. Then, the database $(D \times B)|_A$ is a sub-structure of $D_2$, the structure obtained from $D \times B$ by removing all facts $r(\abf)$ for which $\pi|_{\abf}$ is not injective. Thus, $h$ is a homomorphism from $D$ to $D_2$, or $D^+ \to D_2^+$. As $D^+ \models Q$, it follows that $D_2^+ \models Q$. 
		\end{proof}

%

We prove now the counter-positive of direction `$1 \Rightarrow 3$' of Theorem~\ref{thm:unbArityGDLog}. Let us assume that $\class{D}_{\class{Q}}$ has unbounded semantic submodular width. As all $D \in \class{D}_{\class{Q}}$ are cores and cores witness semantic submodular width \cite{DBLP:conf/ijcai/ChenGLP20}, $\class{D}_{\class{Q}}$ must have unbounded submodular width. Then, according to Theorem~\ref{thm:ChenUnbCSPs} it cannot be evaluated in FPT, and based on our reduction from Theorem~\ref{thm:pCSPtopOMQ}, $\text{p-OMQ}(\class{Q})$ cannot be evaluated in FPT either.

To establish direction `$3 \Rightarrow 2$' of Theorem~\ref{thm:unbArityGDLog}, we show the following: 

\begin{lemma}
	\label{lem:connection_div_cover}
	Let $Q=(\Omc, \Sbf, q)$ be an OMQ from $(\text{GDLog}, UCQ)$ and $Q_c=(\Omc, \Sbf, q_c)$ be its cover. Then, $\mn{SMW}(Q_c) \leq  \mn{SMW}(\class{D}_Q)$. 
	\end{lemma}

\begin{proof}
	
	We show that for every CQ $p_c$ in $q_c$ there exists a database in $\class{D}_Q$ of higher or equal submodular width. Every such CQ is of the form $h(p) \cup S$, for some $D^+ \in \class{D}^+_Q$ and $h$ a homomorphism from some CQ $p$ in $q$ to $ \chase{\Omc}{D^+}$, where $S$ is a set of atoms such that $D \subseteq S \subseteq D^+$, and every atom $r(\abf)$ from $S \setminus D$ has the property that $\abf$ is a maximal guarded set in $D^+$. The latter follows from minimality requirement on $S$ in Definition~\ref{def:cover} as well as from Point (2b) of the same definition. 
	As every atom from $h(p)$ is guarded by some atom in $S$, it follows that $\mn{SMW}(p_c)=\mn{SMW}(S)$. We will show that: $\mn{SMW}(S) \leq \mn{SMW}(D)$.
	
	Assume $\mn{SMW}(D)=k$. Also, let $f:2^{\adom{S}} \to \mathbf{R}_{\geq 0}$ be a monotone submodular edge-dominated function w.r.t. $S$. As $D \subseteq S$, the restriction of  $f$ to $2^{\adom{D}}$, $f'$, has the same properties, in particular it is edge-dominated  w.r.t. $D$. Thus, there exists a tree decomposition $\delta=(T, \chi)$ of $H_D$, with $T=(V, E)$ such that $f(\chi(t)) \leq k$ for all $t \in V$. This can be extended to a tree decomposition  $\delta'=(T', \chi')$ of $H_{D^+}$, with $T'=(V', E')$ by adding the tree decompositions corresponding to the guarded unravelings used in the construction of $D^+$ to $\delta$. An important property of $\delta'$ which is inherited from the property of the tree decompositions of the guarded unravelings added to $D$ is that for every $v' \in V' \setminus V$, $\chi'(v')$ is guarded in $D^+$. 
	
	As $S \subseteq D^+$, we can obtain from $\delta'$ a tree decomposition $\delta_S=(T_S, \chi_S)$ of $S$, with $T_S=(V_S, E_S)$,  by projecting away the constants from $\adom{D^+} \setminus \adom{S}$ and deleting nodes $v' \in V'$ for which $\chi'(v)=\emptyset$. Note that $V \subseteq V_S$ (nodes pertaining to constants from $D$ are not deleted). We will then update edges in $E'$ to re-link disconnected parts of the resulting decomposition. Furthermore, due to the fact that every atom $r(\abf)$ from $S \setminus D$ has the property that $\abf$ is a maximal guarded set in $D^+$, for every node $v \in V_S \setminus V$, $\chi_S(v)$ is still a guarded set in $D^+$. 
	
	We look at the $f$-cost of $\delta_S$: for every node $v \in V_s \cap V$, $f(\chi_S(v))=f'(\chi(v))$, thus $f(\chi_S(v)) \leq k$. On the other hand, for every node $v \in V_s \setminus V$, $f(\chi_S(v)) \leq 1$ as $\chi_S(v)$ is guarded and $f$ is edge-dominated. Thus, $\mn{SMW}(S) \leq k=\mn{SMW}(D)$.

	\qed
	\end{proof}

We next show the fpt upper bound, i.e. direction `$2 \Rightarrow 1$' of Theorem~\ref{thm:unbArityGDLog}: if $\class{Q}_{c}$ has bounded SMW, then, $\text{p-OMQ}(\class{Q})$ is fixed-parameter tractable. The result follows from the fact that for an OMQ $Q$, $Q_c$ is computable from $Q$ (Lemma~\ref{lem:coverComput}) and the following lemma (which concerns also OMQs based on GTGDs):


\begin{lemma}
	\label{lem:FPTsyntactic}
	Let $Q$ be an OMQs from (GTGD, UCQ) of bounded submodular width. Then, $\text{p-OMQ}(Q)$ is fixed-parameter tractable.  	
\end{lemma}

\begin{proof}
	We use some results from \cite{BDFLP-PODS20}. Let \mn{L} be the language of \emph{linear TGDS}, i.e. TGDs which contain only one atom in the body. For an OMQ $Q=(\Omc, \Sbf, q)$ from (GTGD, UCQ), and a database $D$,  Lemma A.3 from  \cite{BDFLP-PODS20}, shows how to construct in $\FPT$ another OMQ $(\Omc^*, \Sbf^*, q)$ from (\mn{L}, UCQ) and a database $D^*$ such that $D \models Q$ iff $D^* \models Q^*$. Then, Lemma A.1 from  \cite{BDFLP-PODS20} shows that deciding whether $D^* \models Q^*$ can be done by considering a finite portion of $\chase{\Omc^*}{D^*}$, which again can be computed in FPT. Thus, $D \models Q$ iff $D' \models q$, where $D'$ is a database which can be computed in FPT. Thus, assuming that $q$ has bounded submodular width, according to Theorem~\ref{thm:ChenUnbCSPs}, deciding whether $D' \models q$ is in FPT as well. 

	\end{proof}
	

\subsection{Results for (GTGD, UCQ)}

The main result which we show in this section is as follows.
 
\begin{theorem} [Main Result 1]
	\label{thm:unbArityGTGD}
	Let $\class{Q}$ be a r.e. class of OMQs from $(\text{GTGD}, \text{UCQ})$. Under the Exponential Time Hypothesis, $\class{Q}$ has bounded semantic submodular width  iff $\class{Q}$ is fixed-parameter tractable.
\end{theorem}

For every OMQ $Q \in  (\text{GTGD}, \text{UCQ})$, there exists an OMQ $Q' \in (\text{GDLog}, \text{UCQ})$,  \emph{the existential rewriting of $Q$}, such that $Q \equiv Q'$ \cite{BDFLP-PODS20}. This allows us to concentrate on OMQs from (\text{GDLog}, \text{UCQ}) for the lower bound: if $\class{Q}$ can be evaluated in fpt, then so can $\class{Q'}$, the class of existential rewritings of OMQs from $\class{Q}$. Then, according to Theorem~\ref{thm:unbArityGDLog}, the class of covers of OMQs from $\class{Q'}$, $\class{Q}'_c$, has bounded submodular width. But, $\class{Q} \equiv \class{Q'}_c$, and thus, $\class{Q}$ has bounded semantic submodular width. 

We turn our attention to the other direction of the theorem. In this case we know that there exists a class of OMQs $\class{Q}_k$ of bounded submodular width, such that $\class{Q} \equiv \class{Q}_k$, which according to Lemma~\ref{lem:FPTsyntactic} could be evaluated in fpt, but we do not know how to compute it. We start by having a looking at the notions we introduced for the syntactic characterization of OMQs from (GDLog, UCQ). A natural question is whether they are semantic, e.g. do two equivalent OMQs have identical sets of qi databases or characteristic databases? If that would be the case, one could obtain the result from Theorem~\ref{thm:unbArityGTGD} for OMQs from  (GDLog, UCQ) straightaway: the class of characteristic databases $\class{D}_{\class{Q}_k}$ for the class of OMQs $\class{Q}_k$ would have bounded submodular width, and thus so would the class of characteristic databases $\class{D}_{\class{Q}}$ for the class of OMQs $\class{Q}$, and the according to Theorem~\ref{thm:unbArityGDLog}, $\class{Q}$ could be evaluated in fpt. However, as the following example shows, this is not the case:

\begin{example}  
	\label{ex:differentCharDB}
	Let $Q_1=(\Omc_1, \Sbf, q_1)$ and $Q_2=(\Omc_2, \Sbf, q_2)$ be the following two OMQs:
	\begin{center}	 
		\begin{tabular}{lllllll}
			$\Omc_1=\{R(x,y) \to A(x)\}$ &&&$\Sbf=\{R\} $&&&$q_1=\exists x \ A(x)$\\ 
			$\Omc_2=\emptyset$&&& $\Sbf=\{R\} $&&& $q_2= \exists x,y \ R(x,y)$
		\end{tabular} 
	\end{center}	 
	Also, let $D=\{R(a,a)\}$ be an $\Sbf$-database. It can be checked that $Q_1 \equiv Q_2$, $D \models Q_1$, and $D \models Q_2$. However, $D$ is qi w.r.t. $Q_1$, but not w.r.t. $Q_2$. Also, $D \in  \class{D}_{Q_1} \setminus \class{D}_{Q_2}$. 
\end{example}

Still, when considering two equivalent OMQs $Q_1$ and $Q_2$, it is possible to construct another equivalent OMQ $Q_{1,2}$ based on one of the original OMQs and the intersection of the sets of (extended) characteristic databases of the two OMQs, similar to the way the notion of cover has been constructed in Section~\ref{sec:cover} based on the set of extended characteristic databases. This will allow us to prove direction ``$\Rightarrow$'' of Theorem~\ref{thm:unbArityGTGD}:

\begin{proof}

As $\class{Q}$  has bounded semantic submodular width, there exists $k>0$ and a class of OMQs $\class{Q}_k$ of OMQs from $(\text{GTGD}, \text{UCQ})$ of submodular width at most $k$ such that $\class{Q} \equiv \class{Q}_k$. By Lemma~\ref{lem:FPTsyntactic}, $\class{Q}_k$ can be evaluated in FPT, and thus so can $\class{Q}'_k$, the class of its existential rewritings. As $\class{Q}'_k \subseteq (\text{GDLog}, \text{UCQ})$, from Theorem~\ref{thm:unbArityGDLog}, it follows that the class of characteristic databases $\class{D}_{\class{Q}'_k}$ also has bounded submodular width. Let $k'$ be such a bound. 

Let $Q'$ be an OMQ from $\class{Q}'$ and $Q'_k$ be an OMQ from $\class{Q}'_k$ such that $Q' \equiv Q'_k$. Also, let $\class{D}^+_{\cap}=\class{D}^+_{Q'} \cap \class{D}^+_{Q'_k}$. We  construct a new OMQ $Q'_{c, \cap}$ based on $\class{D}^+_{\cap}$ and $Q'$ similarly to the way the cover OMQ $Q'_c$ of $Q'$ is constructed based on $\class{D}^+_{Q'}$. At the same time, it is possible to show that $Q'_{c, \cap} \equiv Q'$, and furthermore, if $Q'_{c, \cap}$ is of the form $(\Omc, \Sbf, q'_{c, \cap})$ and $Q'_c=(\Omc, \Sbf, q'_c)$, each CQ in $q'_{c, \cap}$ is also a CQ in $q'_c$. Again, similarly to the result from Lemma~\ref{lem:connection_div_cover} it is possible to show that $\mn{SMW}(Q'_{c, \cap}) \leq \mn{SMW}(\class{D}^+_{\cap})$. As $\class{D}^+_{\cap} \subseteq \class{D}_{Q'_k}$, it follows that $\mn{SMW}(\class{D}^+_{\cap}) \leq k'$, thus $\mn{SMW}(Q'_{c, \cap}) \leq k'$. 
Let $Q'_{c, k'}$ be the OMQ obtained from $Q'_c$ by retaining only those CQs from $q'_c$ which have SMW $\leq k'$. By construction, $Q'_{c, \cap} \subseteq Q'_{c, k'} \subseteq Q'_c$, $Q'_{c, \cap} \equiv Q'_c$, and thus 
$Q'_{c, k'} \equiv Q'$. As $\mn{SMW}(Q'_{c, k'}) \leq k'$, it can be evaluated in fpt. At the same time, $Q'_{c, k'}$ is computable from $Q'$, and thus $Q'$, and subsequently $Q$, can be evaluated in fpt. Thus, so can the corresponding classes $\class{Q'}$ and $\class{Q}$.

\qed
\end{proof}

Finally, we obtain as a corollary of Theorem~\ref{thm:unbArityGDLog} and Theorem~\ref{thm:unbArityGTGD} the following result concerning covers which was already anticipated in Section~\ref{sec:cover}:

\begin{corollary}
	\label{cor:semWitnessCovers}
	Let $\class{Q}$ be a class of OMQs from (GDLog, UCQ). Then, $\class{Q}$ has bounded semantic submodular width iff its class of covers $\class{Q}_c$ has bounded submodular width.
	\end{corollary}

The question remains open whether covers are witnesses for semantic submodular width when considered individually, as opposed to witnesses for bounded semantic submodular width when considered as classes, as established in Corollary~\ref{cor:semWitnessCovers}. We leave this open for now. 

\section{Revisiting the Bounded Arity Case}
\label{sec:revisiting}

We revisit the characterization for fpt evaluation of OMQs from (GTGD, UCQ) over bounded arity schemas from~\cite{BDFLP-PODS20}. A class of OMQs $\class{Q}$ is \emph{over bounded arity schemas} if there exists a $k$ such that every schema in some OMQ in $\class{Q}$ contains symbols of arity at most $k$. 

\begin{theorem}[Theorem 5.3, \cite{BDFLP-PODS20}]
	\label{thm:boundArityGTGDChar}
	Let $\class{Q}$ be a r.e. class of OMQs from (GTGD, UCQ) over bounded arity schemas. Assumming $\FPT \neq \W $, the following are equivalent : 
	\begin{enumerate}
		\item $\text{p-OMQ}(\class{Q})$  is fixed-paramater tractable.
		\item $\class{Q}$ has bounded semantic treewidth. 
	\end{enumerate}
	If either statement is false, then $\text{p-OMQ}(\class{Q})$  is $\W$-hard. 
\end{theorem}

 The characterization generalizes a previous result concerning complexity of evaluating OMQs from ($\mathcal{ELHI}_{\bot}$, UCQ) \cite{BFLP19} and
 can be seen as a generalization of Grohe's complexity results regarding the parameterized complexity of uniform CSPs over bounded arity schemas \cite{Grohe07}: 
 
  \begin{theorem}[Theorem 1, \cite{Grohe07}]
  	\label{thm:grohe}
 	Assume that $\FPT \neq \W$. Then for every r.e. class $\class{C}$ of structures of bounded arity the following statements are equivalent: 
 	\begin{enumerate}
 		\item $\text{CSP}(\class{C}, \_)$ is in polynomial time. 
 		\item $\text{p-CSP}(\class{C}, \_)$  is fixed-paramater tractable.
 		\item $\class{C}$ has bounded treewidth modulo homomorphic equivalence. 
 	\end{enumerate}
 	If either statement is false, then $\text{p-CSP}(\class{C}, \_)$  is $\W$-hard. 
 \end{theorem}

 Direction~`2 $\Rightarrow$ 1' of Theorem~\ref{thm:boundArityGTGDChar} is established in \cite{BDFLP-PODS20} using arguments regarding the chase construction and applying the result from \cite{Grohe07} on a finite portion of the chase. As concerns the lower bound (direction `1 $\Rightarrow$ 2'), the proof from \cite{BDFLP-PODS20} is extremely complex: it lifts in  a non-trivial way the fpt reduction from the parameterized $k$-clique problem to parameterized uniform CSPs over  
 bounded arity schemas from \cite{Grohe07}. Lifting the reduction  to the OMQ case required in particular introspection into a certain construction used in the original proof, \emph{the Grohe database},  and modifying that construction. 

Here we sketch how it is possible to establish the results from  Theorem~\ref{thm:boundArityGTGDChar} using Grohe's result as a black box by employing the reduction from Theorem~\ref{thm:pCSPtopOMQ}. This shows the potential of the reduction for lifting results from the CQ evaluation realm to the OMQ one. 
We start by establishing a counterpart of Main Result 3 for the case of bounded arities OMQs. 

 \begin{theorem}[GDLog Bounded Arity Characterization]
 	\label{thm:boundedArityGDLog}
 	Let $\class{Q}$ be a r.e. class of OMQs from $(\text{GDLog}, \text{UCQ})$ over bounded arity schemas. 	Assuming that $\FPT \neq \W$: 
 	\begin{enumerate}
 		\item $\text{p-OMQ}(\class{Q})$ is fixed-parameter tractable iff 
 		\item $\class{Q}_{c}$ has bounded tree-width iff 
 		\item $\class{D}_{\class{Q}}$ has bounded tree-width. 
 	\end{enumerate}
 	If either statement is false, then $\text{p-CSP}(\class{C}, \_)$  is $\W$-hard.
 \end{theorem}

The strategy to prove Theorem~\ref{thm:boundedArityGDLog}
is similar to the one used to prove Theorem~\ref{thm:unbArityGDLog}, except that this time Theorem~\ref{thm:grohe} from \cite{Grohe07} is used as a blackbox, as opposed to the unbounded arity case, where the results for uniform CSPs over unbounded arity schemas from Theorem~\ref{thm:ChenUnbCSPs} were used as a blackbox. This is the case both for showing the upper bound, direction `2 $\Rightarrow$ 1' of the theorem, and the lower bound,  direction `1 $\Rightarrow$ 3' of the theorem. In the latter case, we use again our reduction from parameterized uniform CSPs to parameterized OMQs. What still needs to be shown, is the connection between the treewidth of the cover of an OMQ and the treewidth of the set of characteristic databases of the same OMQ. It follows straightaway from the proof of Lemma~\ref{lem:connection_div_cover} that:

\begin{lemma}
	\label{lem:treewidth_div_cover}
	For $Q=(\Omc, \Sbf, q)$ an OMQ from $(\text{GDLog}, UCQ)$ with schema arity at most $r$, and $Q_c=(\Omc, \Sbf, q_c)$ its cover, it is the case that $\mn{TW}(Q_c) \leq \mn{max}(r, \mn{TW}(\class{D}_Q))$. 
\end{lemma}

Direction `3 $\Rightarrow$ 2' of Theorem~\ref{thm:boundedArityGDLog} follows from Lemma~\ref{lem:treewidth_div_cover}. By using Theorem~\ref{thm:boundedArityGDLog}, we can then retrieve the results from Theorem~\ref{thm:boundArityGTGDChar} similarly as we did for Theorem~\ref{thm:unbArityGTGD}. 



\section{Conclusions and Future Work}

In this work, we characterized the fpt border for evaluating classes of parameterized OMQs based on guarded TGDs and UCQs in the unbounded arity case. For ontologies expressed in GDLog, we provided a syntactic characterization based on new constructions, namely sets of characteristic databases and covers of an OMQ. On the way to establish this result, we introduced an fpt reduction from evaluating parameterized uniform CSPs to evaluating parameterized OMQs.

The reduction enables the lifting of results from the CSP world to the OMQ one in a modular fashion. To further showcase this, we revisited the case of OMQs over bounded arity schemas, previously addresed in \cite{BDFLP-PODS20}. For classes of such OMQs from (GDLog, UCQ) we  established a new syntactic characterization of the tractability border in terms of covers and characteristic databases, while for classes of OMQs from (GTGD, UCQ), we showed how the reduction enabled a much simpler proof for the semantic characterization from \cite{BDFLP-PODS20}. 

Here we only considered Boolean OMQs, as the corresponding results for CQ evaluation in the unbounded arity case have also only been established in the Boolean case. As future work, we plan to extend our work to the non-Boolean case. Many results from the database world concerning efficiency of performing tasks like counting \cite{chen_et_al:LIPIcs:2015:4980}, enumeration \cite{berkholz_et_al:LIPIcs:2019:11002,carmeli_et_al:LIPIcs:2018:8598}, and so on, use structural measures on the queries similar to the ones involved in the characterizations for evaluation. Thus, by generalizing our constructs to the non-Boolean case, depending on which structural measures are preserved when transitioning from sets of characteristic databases to covers of OMQs, we might be able to lift such results to the OMQ world.

\bibliographystyle{plain}
\bibliography{references_unb}
\end{document}